\documentclass[11pt]{article}

\usepackage{subcaption} % Use subcaption package
\usepackage{amsthm}
\newtheorem{theorem}{Theorem}

\usepackage[margin=1in,lmargin=1.25in,rmargin=1.25in]{geometry}
\usepackage{graphicx}%
\usepackage{multirow}%
\usepackage{amsmath,amssymb,amsfonts}%
\usepackage{mathrsfs}%
\usepackage[title]{appendix}%
\usepackage{xcolor}%
\usepackage{textcomp}%
\usepackage{manyfoot}%
\usepackage{booktabs}%
\usepackage{algorithm}%
\usepackage{algorithmicx}%
\usepackage{algpseudocode}%
\usepackage{listings}%
\usepackage{microtype}
\usepackage{cite}
\usepackage{url}
\usepackage[colorlinks,linkcolor=black,citecolor=black,urlcolor=black]{hyperref}

\usepackage{caption}
\usepackage{etoolbox}

\AtBeginEnvironment{thebibliography}{
    \footnotesize  % Set bibliography to footnote size
    \setlength{\itemsep}{0pt}  % No space between items
    \setlength{\parsep}{0pt}   % No paragraph spacing within items
    \setlength{\parskip}{0pt}  % Just in case
}
         
\usepackage[tiny,compact]{titlesec}
\usepackage{todonotes}

\usepackage{tabularx}

\usepackage{authblk}

\newcommand{\mat}[1]{\boldsymbol{#1}}
\newcommand{\mX}{\mat{X}}
\newcommand{\vd}{\mat{d}}

\algtext*{EndIf}     % Remove \EndIf
\algtext*{EndWhile}  % Remove \EndWhile
\algtext*{EndFor}    % Remove \EndFor
\algtext*{EndProcedure} % Remove \EndProcedure (if you're using procedures)

\newcommand{\eps}{\varepsilon}

\title{
A spatial hypergraph model to smoothly interpolate between pairwise graphs and hypergraphs to study higher-order structures
}
\author[1]{Omar Eldaghar}
\author[2]{Yu Zhu}
\author[2]{David F. Gleich}
\affil[1]{Department of Mathematics, Purdue University}
\affil[2]{Department of Computer Science, Purdue University}

\date{}

\begin{document}

\maketitle

% \section{Submission Information}
% \paragraph{Submission Deadline: February 28, 2025}
% Slides available at: \url{https://docs.google.com/presentation/d/1t-Q_r4Jq-ve6vhocp1NjQW-A85vn-xwbqytPVFPLmqY/edit?usp=sharing}

% TODO list:
% 0. Write an abstract
% 0. Revise introduction
%    Key pitch: a model that interpolates from pairwise to hyperedge structure.  
%    as a way / control to study the impact of higher-order struture. 
% Fix refs in intro (we need 20-30 refs...)
% 1. Clustering coefficients on the multi-edge graph
% 1. (for Yu) Writeup case study section on modeling networks.
%    Omar's node2vec idea implemented by Yu.
%    For graphs, use node2vec. 
% 2. Omar's crazy interpolation idea (Watts-Strogatz)
%     watts strogatz as a mixture of a random sampled hyperedge
%    and something with local structure
% 2. Other ditributions of points?
%     SBM like points from a mixture of Gaussins on a k-1 dim. simplex
% 1. PPR Diffusion 
% 1. Fix up/finish where's here... 

\begin{abstract}
We introduce a spatial graph and hypergraph model that smoothly interpolates between a graph with purely pairwise edges and a graph where all connections are in large hyperedges. The key component is a spatial clustering resolution parameter that varies between assigning all the vertices in a spatial region to individual clusters, resulting in the pairwise case, to assigning all the vertices in a spatial region to a single cluster, which results in the large hyperedge case. An important outcome of this model is that the spatial structure is invariant to the choice of hyperedges. Consequently, this model enables us to study clustering coefficients, graph diffusion, and epidemic spread and how their behavior changes as a function of the higher-order structure in the network with a fixed spatial substrate. We hope that our model will find future uses to distill or explain other behaviors in higher-order networks. 

% Including higher-order methods in network analysis has increased our understanding of relational phenomena by modeling data and relationship in their natural form.
% Moreover, their inclusion in modeling dynamics such as epidemics has shown important deviations from pairwise models.
% However, many pairwise notions that seem natural and intuitive can generalize in non-trivial ways when moving to higher-order data.
% In order to better understand the differences between pairwise and higher-order effects we introduce a random hypergraph model that interpolates between spatial pairwise graph and hypergraphs.
% We give several case-studies highlighting the utility of our model.
% Our model will enable researchers to better understand higher-order structure and metrics. 
\end{abstract}

\section{Introduction}

Parametric graph and hypergraph models have played an important role in network science and complex systems research --- from the Watts and Strogatz model of local clustering and graph diameter~\cite{watts1998collective}, to Kleinberg's study of small world routing~\cite{kleinberg2000small}, to the famous mixing parameter in the Lancichinetti, Fortunato, and Radicchi (LFR) stochastic block model~\cite{lancichinetti2008benchmark}. More generally, parametric variation in complex systems helps to identify different regimes of behavior and often phase transitions among them. Two examples of this include connectivity in a simple uniform random graph (Erd\H{o}s--R\'enyi--Gilbert model)~\cite{er1959, er1960, er1961, gilbert1959random} and synchronization in the Kuramoto oscillation~\cite{kuramoto1975self,Strogatz2000}. The key feature of these models is that they allow us to study a system as only one aspect varies. This is important because a single well-defined notion in the pairwise setting can have several orthogonal generalizations when moving to higher-order structures. Studying these generalizations in a controlled scenario enables us to appreciate and understand differences in the results. 

In this paper, we introduce a spatial model of graphs and hypergraphs that enables parametric variation from a purely pairwise edge behavior to the case where all nodes are involved in only hyperedges.
\footnote{This model was originally proposed in a conference paper by us~\cite{Eldaghar-2024-spatial-hypergraphs}. We discuss some slight differences between the conference model and the model in this paper shortly.} 
Crucially, this can be done while retaining the same spatial graph substrate. This enables us to study changes due to the hyperedge structure alone, which represents a unique capability among hypergraph models. 

The proposed model is simple and scalable. We randomly assign points to each vertex and also assign a number of neighbors, which is usually sampled from a random distribution. In the pairwise case, they directly connect to this number of nearest neighbors. However, in the general case, we run a clustering algorithm on the spatial connections among the points within this spatial region. The idea is that if points represent some latent similarity space, then nearby points will reflect similar features. Consequently, we use the spatial clusters in these regions to induce hyperedges. By varying the spatial cohesion, we can adjust the presence of pairwise edges compared with hyperedges. We discuss the model formally in Section~\ref{sec:model}. One challenge was scaling the spatial cohesion parameter to preserve scaling as we vary the dimension of the space from which points are drawn. A different feature of this model is that much of the connectivity among nodes is invariant to the choice of edges versus hyperedges, a result we formalize in Theorem~\ref{thm:connectivity}. 

We then explore how this model enables us to study the impact of higher-order structures in hypergraphs. Our first study is on clustering coefficients (Section~\ref{sec:clustering}). There are many different types of clustering coefficients in hypergraphs. We study a few of the simplest and most common including both unweighted and weighted clustering coefficients of the clique expansion as well as bipartite clustering coefficients of the node-hyperedge incidence matrix. In what is a small surprise, scaling from medium-sized to large-sized hyperedges reduces the global clustering coefficient of the projected graph (Figure~\ref{fig:cc-figure}). This occurs because adding new projected hyperedges can cause the number of length-2 paths, or wedges, to grow much faster than one might expect. This study also shows how the behavior of different clustering coefficients varies quite substantially as we interpolate from pairwise graphs to hypergraphs, which suggests that results about clustering coefficients in hypergraphs may not be robust to changing the type of clustering coefficient used. 

The next study is in terms of diffusion (Section~\ref{sec:diffusion}). We use a personalized PageRank diffusion in hypergraphs~\cite{Liu-2021-localhyper}. In this case, the behavior of the diffusion is governed by the spatial substrate underlying the network. Consequently, we see little difference in the behavior of the diffusion as we move from the pairwise case to the hypergraph case. The choice to include this study is meant to show that the model behaves as expected when higher-order structure may not impact the underlying physics.  

The final study is on epidemic spread (Section~\ref{sec:epidemics}). In this case, there is tremendous uncertainty about the impact of higher-order structure and, indeed, a variety of mixed results in the literature. For example, the addition of higher-order structures and group-level spreading can greatly alter the stability of epidemic thresholds by inducing a region of bistability~\cite{iacopini2019simplicial} not seen in traditional pairwise models.
Moreover, heterogeneity can also play competing roles in pairwise and higher-order structures~\cite{landry2020effect} to accelerate or inhibit spreading.
In this case, we wish to study epidemic spread in a model that attempts to mimic an airborne virus in the presence of ventilation. In this scenario, large hyperedge interactions require spaces with additional ventilation, which corresponds to a dilution effect of infectious aerosols. Of note, we find that the impact of higher-order structure varies with the epidemic parameters in non-intuitive ways in this scenario (Section~\ref{sec:varying-impact}).

This paper extends a previous introduction of these ideas from the same authors~\cite{Eldaghar-2024-spatial-hypergraphs}. Key differences in this greatly expanded version include (i) a discussion of the model beyond two-dimensional spatial graphs, (ii) a theoretical characterization of the connectivity of the model, and  (iii) studies of the hypergraph model in terms of clustering coefficients as well as (iv) graph diffusion. Finally, the epidemic study includes a more detailed analysis of the specific hypergraph mechanisms underlying the differences observed in~\cite{Eldaghar-2024-spatial-hypergraphs}.  We discuss additional related work in the space of random geometric graphs, random geometric hypergraphs, and random geometric simplicial complexes in Section~\ref{sec:related}. The ability to easily interpolate between a pairwise graph and a hypergraph appears to be a unique feature of this model.

\section{Model Description}
\label{sec:model} 

\begin{table}
\centering 
\begin{tabularx}{0.75\linewidth}{rX}
\toprule 
Symbol & Interpretation \\
\midrule 
$\mX$ & the set of all points \\
$D(\cdot, \cdot)$ & the distance metric \\  
$\mX_v$ & the coordinate of a specific point $v$ \\ 
$d_v$ & the degree associated with point $v$ \\ 
$\boldsymbol{d}$ & the set of all degrees for all vertices \\ 
$N(v)$ & the set of $d_v$ nearest neighbors of $v$ under distance $D$ \\ 

\bottomrule 
\end{tabularx}
\caption{A summary of our notation} 
\label{}
\end{table}

The model we propose is simple, fast, and flexible. It begins with a set of points in a space $\mX$ along with a distance metric $D$. In all of our studies, these points are sampled from the $d$-dimensional unit-cube $[0,1]^d$ uniformly at random, although in principle a different distribution can be used. For each point $v$, we give it a radius of influence, $d_v$, expressed as a number of nearest neighbors, which we call the \emph{degree}.  Note that this in a mild abuse of notation, because in the final graph construction, the degree of the node is typically larger than $d_v$, although $d_v$ is a lower bound on the degree.  We typically sample values of $d_v$ from a log-normal distribution.

Let $\mX_v$ be the coordinates of a point $v$ and $d_v$ be the associated degree. In a standard spatial nearest-neighbor graph, we would connect node $v$ to the nearest $d_v$ neighbors by adding edges for each neighbor. Our model is based on this setup. However, we wish to \emph{cluster} the points within $v$'s radius of influence. Let $N(v)$ be the set of $d_v$ nearest neighbors for point $v$, and let $r_v$ be the distance to the $d_v$th nearest neighbor, formally, $r_v = \max_{u \in N(v)} D(v, u)$. The goal is to cluster the set of points in $N(v)$ and form edges or hyperedges based on these \emph{clusters} rather than individual points. (To be completely clear, we do not consider $v$ in the set of points we cluster.)  We illustrate this process in Figure~\ref{fig:model-description}. Formally, for each cluster of points in $N(v)$ we create a hyperedge consisting of all points in the cluster combined with the original point $v$. The inspiration for this idea is that we would have a \emph{group interaction} among $v$ and points that are all themselves \emph{close} within its region of influence. 

\begin{figure}
    \centering
    \includegraphics[width=0.65\linewidth]{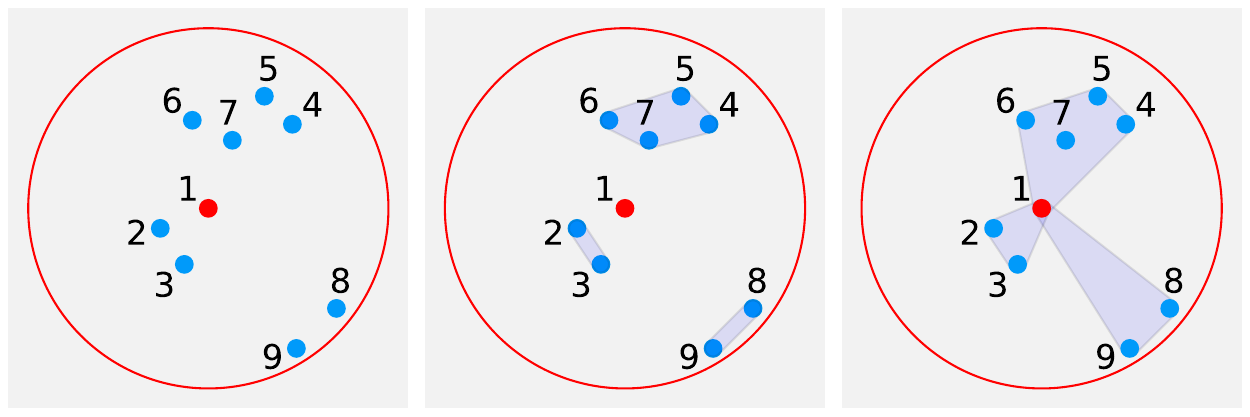}
    \caption{How hyperedges are formed in our model. Eight nearest neighbors are computed for node $1$ (leftmost plot) which are then clustered into $3$ clusters (middle plot). Finally, each of the clusters serves as a separate group interaction for node $1$ and they become hyperedges.}
    \label{fig:model-description}
\end{figure}

The key feature of the model is that we can control the behavior of the graph by controlling the behavior of the clustering function. Suppose that each point in $N(v)$ is clustered into a \emph{separate cluster}. Then we simply recover the pairwise nearest neighbor graph among the points. Alternatively, suppose that $N(v)$ is clustered into a \emph{single cluster}. Then we recover a geometric hypergraph construction where all edges are hyperedges (unless $d_v = 1$). Thus, by varying the cluster sizes, we can control the extent of hyperedge effects. This idea is illustrated in Figure~\ref{fig:varying-alpha}.

\begin{figure}
    \centering
    \includegraphics[width=0.75\linewidth]{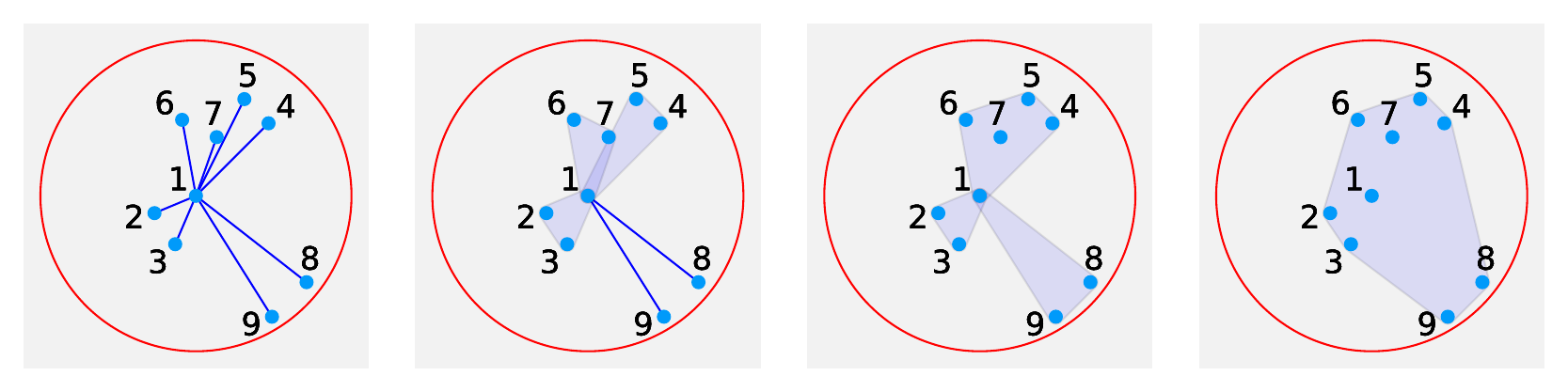}
    \caption{As we vary the number of clusters produced among each node's nearest neighbors, we are able to interpolate between purely pairwise (leftmost plot) and purely higher-order structure (rightmost plot). }
    \label{fig:varying-alpha}
\end{figure}

The model could accommodate any clustering function we desire and, in addition, support attributes on nodes as well. However, for concreteness, we use the DBSCAN clustering algorithm~\cite{ester1996density}. The DBSCAN method depends on a single parameter to control the clustering: $\eps$. The choice of $\eps$ is a distance and governs when two points are considered in the neighborhood of each other, and, in turn, whether or not they might be placed in the same cluster. While we can change $\eps$ to achieve the interpolation from pairwise to hypergraph we want, the effect depends greatly on the local context of each node. Consequently, we wish to develop a simple and interpretable parameter to control the interpolation. 

We introduce a parameter $\alpha$ for this goal. When we set $\alpha = 0$,  we want a pairwise graph -- corresponding to the scenario at the left of Figure~\ref{fig:varying-alpha}. When we set $\alpha = 2$, we want a graph where all edges are complete hyperedges within the region of influence -- corresponding to the scenario at the right of Figure~\ref{fig:varying-alpha}. For $\alpha=1$, we'd like a point in the middle where there are around $\sqrt{d_v}$ clusters. Consequently, we want to build a function
\[ \eps_{\alpha}(r_v, d_v, d) \]
that depends on the number of neighbors $d_v$, the distance enclosed by the region of influence $r_v$, along with the ambient dimension $d$ and produces a value of $\eps$ for DBSCAN to achieve this goal. Figure~\ref{fig:scaling-distance} illustrates the reason why the function needs to scale with both $d_v$ and $r_v$ (and implicitly, the ambient dimension $d$). 

\begin{figure}
    \centering
    \includegraphics[width=0.75\linewidth]{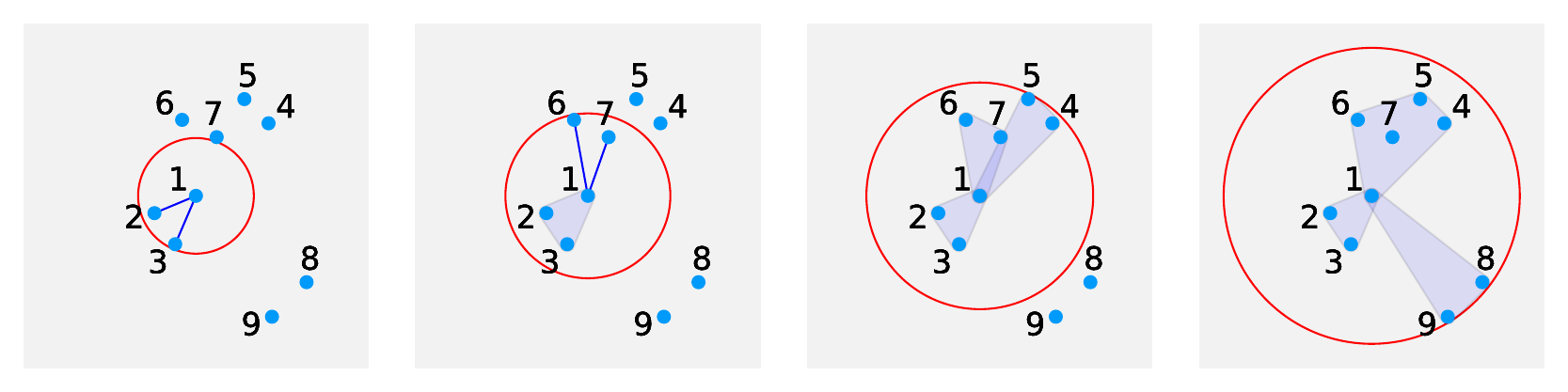}
    \caption{Hyperedges formed around the node $v=1$ as the number of nearest neighbors, $d_v$, increases. We want the neighborhood radius parameter $\eps$ of DBSCAN to scale with $d_v$ and $r_v$ (the maximum distance among the $d_v$ neighbors). To establish a concrete and controllable model, we design a function $\epsilon_\alpha$ that interpolates between individual clusters around each point and a single cluster. This function needs to scale with the parameters of the neighborhood to achieve its aims.}
    \label{fig:scaling-distance}
\end{figure}

We map $\alpha$ to distance parameter $\eps$ in DBSCAN via 
\begin{equation}
    \eps_\alpha(r_v, d_v, d) = \begin{cases} 
    \alpha^{1/d} \frac{r_v}{d_v^{1/2d}}, \qquad\text{ if }\alpha\in [0,1]\\
    \frac{r_v}{d_v^{1/2d}} + (\alpha-1)\left(r_v-\frac{r_v}{d_v^{1/2d}}\right). \quad \alpha\in(1,2]
    \end{cases}
    \label{eq:dbscan-radius}
\end{equation}

 Recall that we want about $\sqrt{d_v}$ clusters when $\alpha=1$. So we want to scale our value of $\eps_1$ to do so. 
The intuition for our choice of $\eps_1$ is that if $B_{r_v}^d$ denotes a $d-$dimensional ball with radius $r_v$, then if we split the volume of this ball,  $Vol(B_{r_v}^d)$, into $\sqrt{d_v}$ equal pieces (ignoring the issue of sphere packing), it would yield pieces of volume $\sqrt{d_v}\eps_1^d \approx r_v^d$ so that $\eps_1 \approx \frac{r_v}{d_v^{1/2d}}$. 
Similarly, if we think about placing a single ball rather than performing DBSCAN, a scaling argument would yield an expected hyperedge size of $\frac{d_v}{vol(B_{r_v}^d)}vol(B_{\epsilon_\alpha}^d) \sim d_v\frac{\epsilon_\alpha^d}{r_v^d}$. 
The choice of $\epsilon_\alpha$ in Equation~\eqref{eq:dbscan-radius} then gives an estimate of linear scaling for $\alpha\in [0,1]$ to account for changes in the dimension $d$.

Our complete spatial hypergraph construction procedure is described in Algorithm~\ref{alg:model}. This shows how we build up a list of hypergraph edges by considering the results of the clustering. We also have our computational codes available at \url{https://github.com/oeldaghar/spatial-hypergraph-epidemics}
that implement this routine. To give an intuition for the resulting graph, a few small samples are shown in Figure~\ref{fig:graph-demo}.

\begin{algorithm} 
\caption{Spatial Hypergraph Model}
\begin{algorithmic}[1]
\Function{SpatialHypergraph}{$\mX$,$\boldsymbol{d}$,$f(\cdot ,\alpha)$} where $\vd$ gives the degree for each node, and $f(\cdot,\alpha)$ is a clustering algorithm with parameter(s) $\alpha$ such that $\alpha = 0$ will cluster into individual pieces, $\alpha = 2$ will cluster into a single group, and $\alpha = 1$ will cluster into about $\sqrt{\text{points}}$ groups
    \State $H \gets [\text{ }]$ \Comment{Initialize an empty list of hyperedges}
    \For{$v = 1:n$} \Comment{for each point in the the set $\mX$}
        \State $N(v)\gets $ $d_v$ nearest neighbors of $\boldsymbol{X}_v$ excluding $\boldsymbol{X}_v$
        \State $\boldsymbol{Y} \gets \{ X_u \text{ for } u \in N(v)\}$ \Comment{Build a subset of points to cluster}
        \State $C_v \gets f(\boldsymbol{Y},\alpha)$ \Comment{run the clustering with parameter $\alpha$}
        \For{$C$ in $C_v$} \Comment{for each cluster in the output}
            \State APPEND($C,v$) \Comment{add vertex $v$ to the cluster before we add it as a hyperedge}
            \State APPEND($H,C$) \Comment{add a new hyperedge to the graph}
        \EndFor
     \EndFor
    \State \Return $H$
\EndFunction
\end{algorithmic}
\label{alg:model}
\end{algorithm}

%%%%%%%%% VISUALIZING GRAPHS GENERATED %%%%%%%%%
\begin{figure}
    \centering
    \begin{subfigure}[h]{0.32\textwidth}
        \centering
        \includegraphics[width=\textwidth]{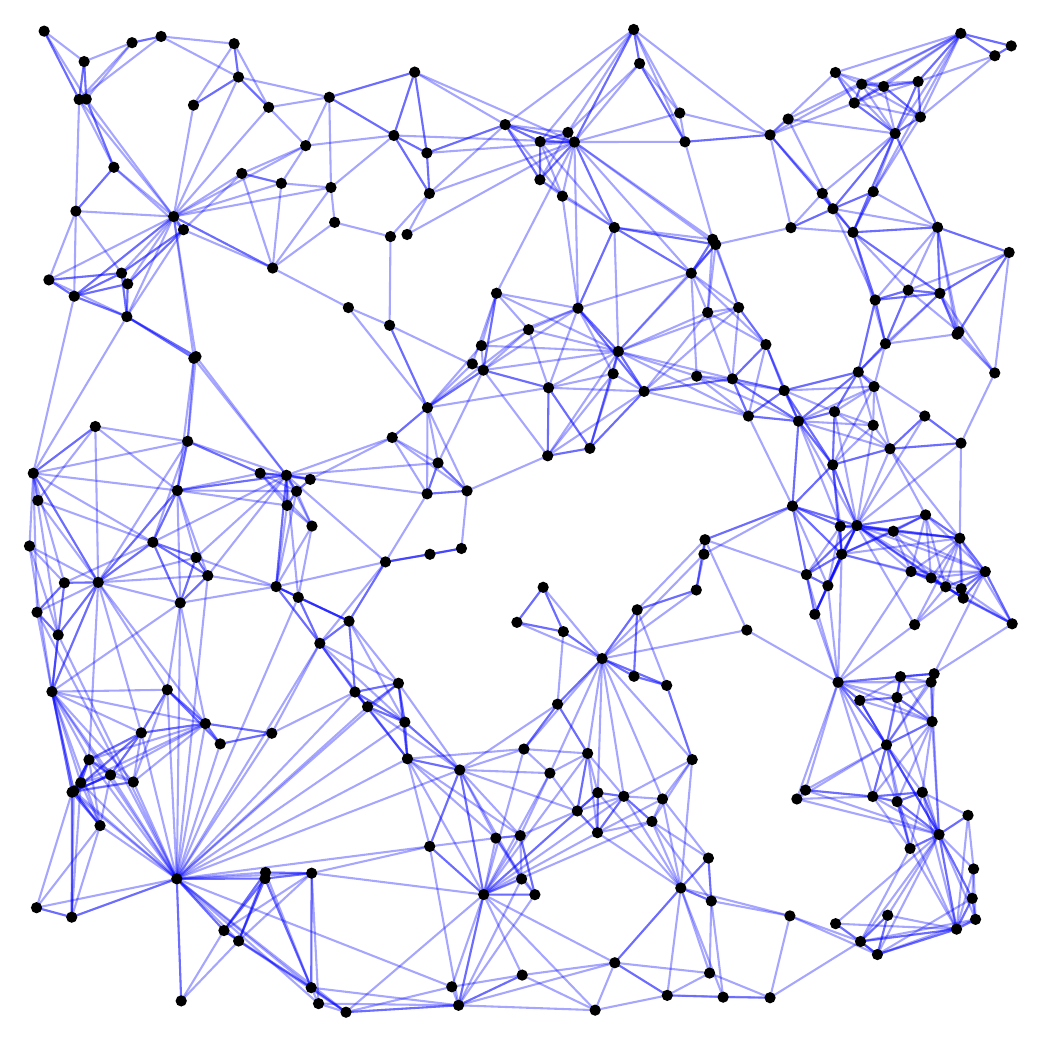}
        \label{fig:graph-demo-1}
    \end{subfigure}
    \begin{subfigure}[h]{0.32\textwidth}
        \centering
        \includegraphics[width=\textwidth]{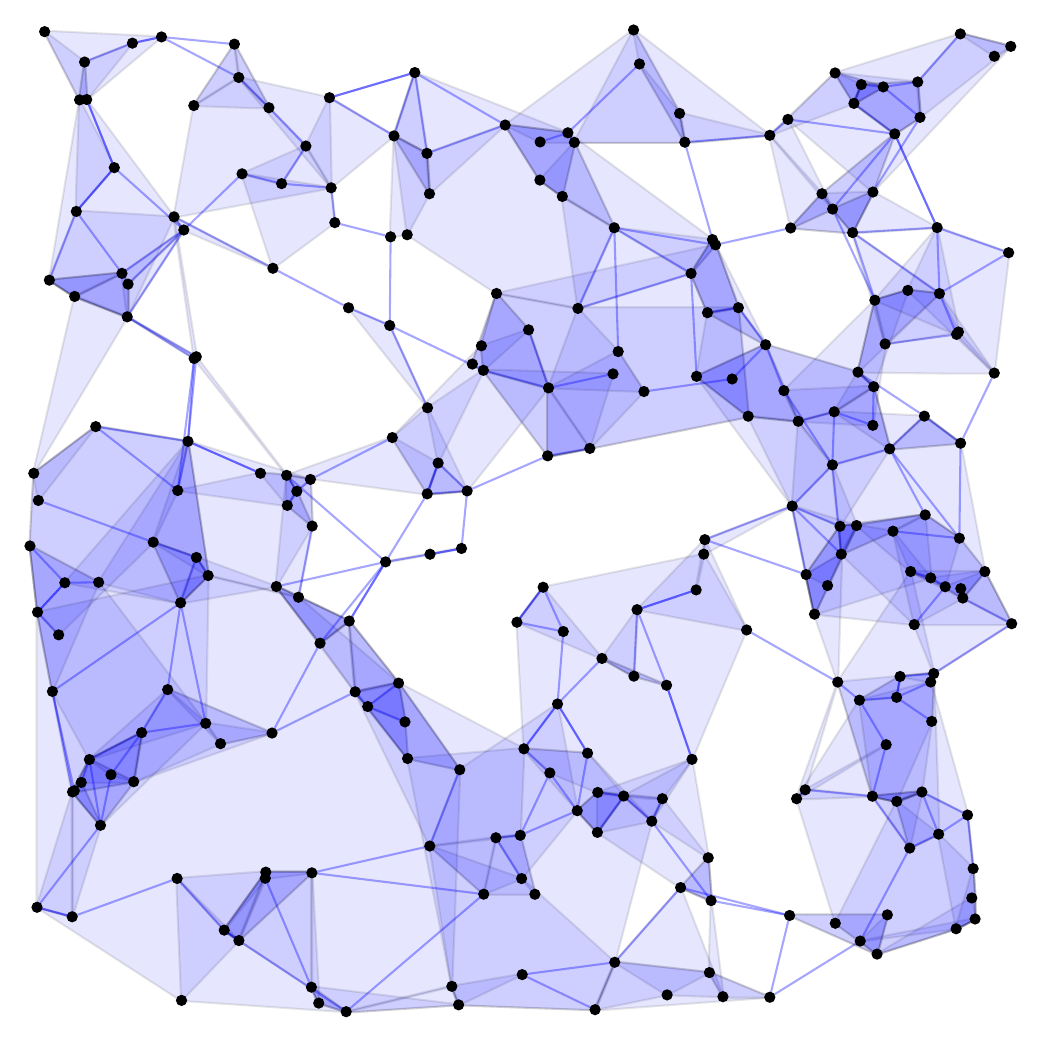}
        \label{fig:graph-demo-2}
    \end{subfigure}
    \begin{subfigure}[h]{0.32\textwidth}
        \centering
        \includegraphics[width=\textwidth]{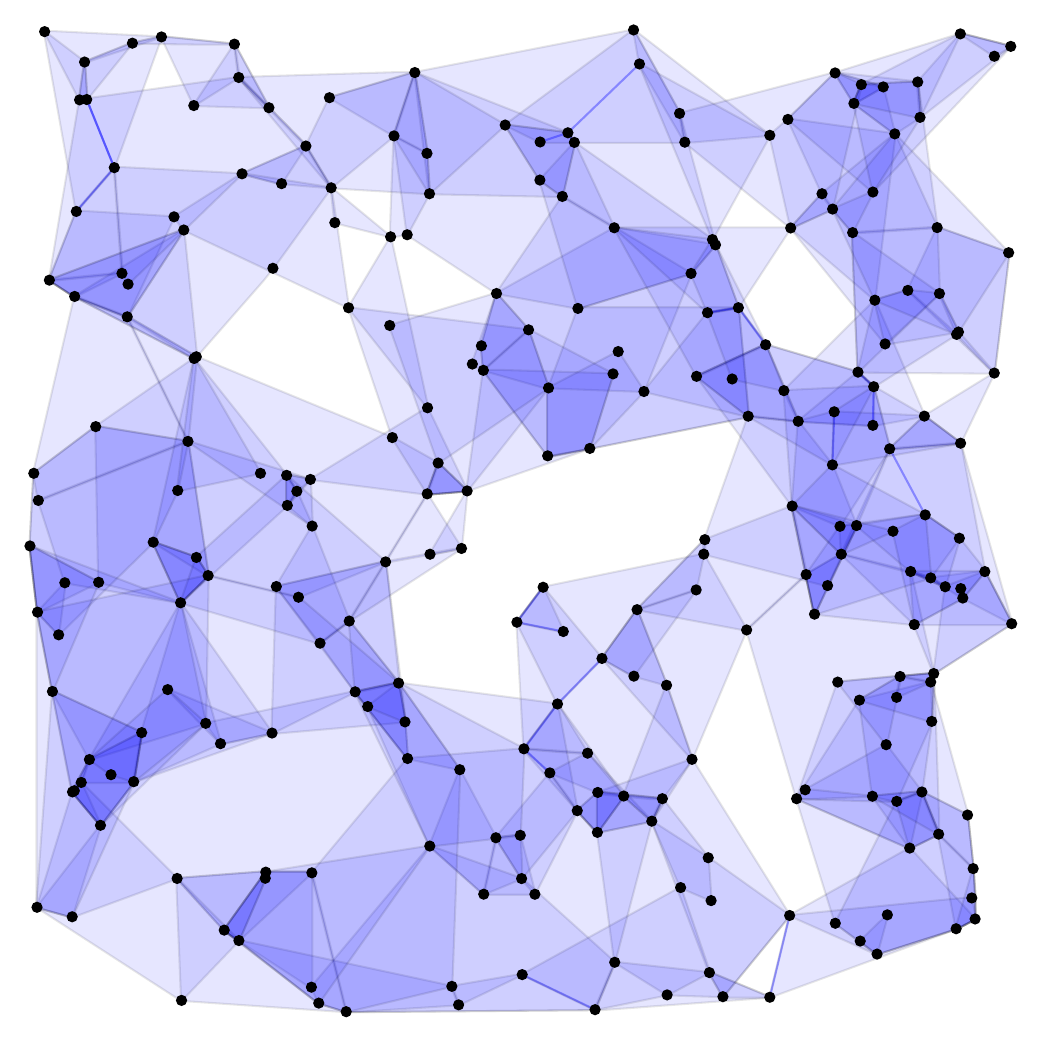}
        \label{fig:graph-demo-3}
    \end{subfigure}
    \caption{Hypergraphs generated on $n=250$ nodes in $d=2$ dimensions for $\alpha = 0, 1, 2$ (left to right) for the same spatial embedding $\boldsymbol{X}$ and specified degrees $\boldsymbol{d}$.}
    \label{fig:graph-demo}
\end{figure}

\subsection{The choice of clustering distance}
\label{sec:model-distance}
This choice of clustering distance $\eps_\alpha$ differs from that in our previous paper~\cite{Eldaghar-2024-spatial-hypergraphs}. In particular, the scaling in~\cite{Eldaghar-2024-spatial-hypergraphs} did not scale with dimension and had a different midpoint $\eps_1$ value. We illustrate the importance of setting this parameter correctly in Figure~\ref{fig:different-alpha-choices}. 

We generate data from our model where $n=10000$  points are sampled from $[0,1]^d$, the degree distribution is sampled from a log-normal with parameters $\mu = \log(3)$ and $\sigma = 1$ (which gives an overall average degree of around 3.5). We fix the geometric information $\boldsymbol{X}$ and node degrees $d_v$ for 25 simulations as we vary the parameter $\alpha$. Put another way, we do not regenerate the spatial information for each distinct value of $\alpha$, and reuse use the same information for an entire sweep through the choices of $\alpha$.

In Figure~\ref{fig:different-alpha-choices}, we compare our choice of clustering distance parameter (left column) with a purely linear scaling  $\eps'_\alpha = \alpha r_v / 2$ (so that $\eps = r_v$ when $\alpha = 2$) and also against  
\begin{equation}
    \eps''_\alpha(r_v, d_v, d) = \begin{cases} 
    \alpha \frac{r_v}{d_v^{1/2d}}, & \alpha\in [0,1]\\
    \frac{r_v}{d_v^{1/2d}} + (\alpha-1)\left(r_v-\frac{r_v}{d_v^{1/2d}}\right). & \alpha\in(1,2]
    \end{cases}
    \label{eq:dbscan-radius-alt}
\end{equation}
The change in the alternative equation is that we scale by $\alpha$ instead of $\alpha^{1/d}$ when $\alpha \le 1$. 

The quantity we measure in the figure is the number of total hyperedges. We expect this to decrease as we increase $\alpha$ because we are forming larger hyperedges as the clusters get larger. We would ideally like this decay to be smooth. The figure shows the smoothest decay for our choice of Equation~\eqref{eq:dbscan-radius}. Linear scaling (in the middle column) looks good in $2$ dimensions, but results in almost no hyperedges until $\alpha$ is larger for higher dimensions. Again, this is expected because randomly spaced points separate out in higher dimensions and become further away. Likewise, we see the same impact in Equation~\eqref{eq:dbscan-radius-alt} with the linear scaling between $\alpha = 0$ and $1$. The \emph{jump} near $\alpha=2$ occurs because of a discontinuity in the behavior of the clustering algorithm as we go from $2$ clusters to $1$ cluster at each vertex, which cannot be smooth.

\begin{figure}[t]
    \centering
    \includegraphics[width=0.9\linewidth]{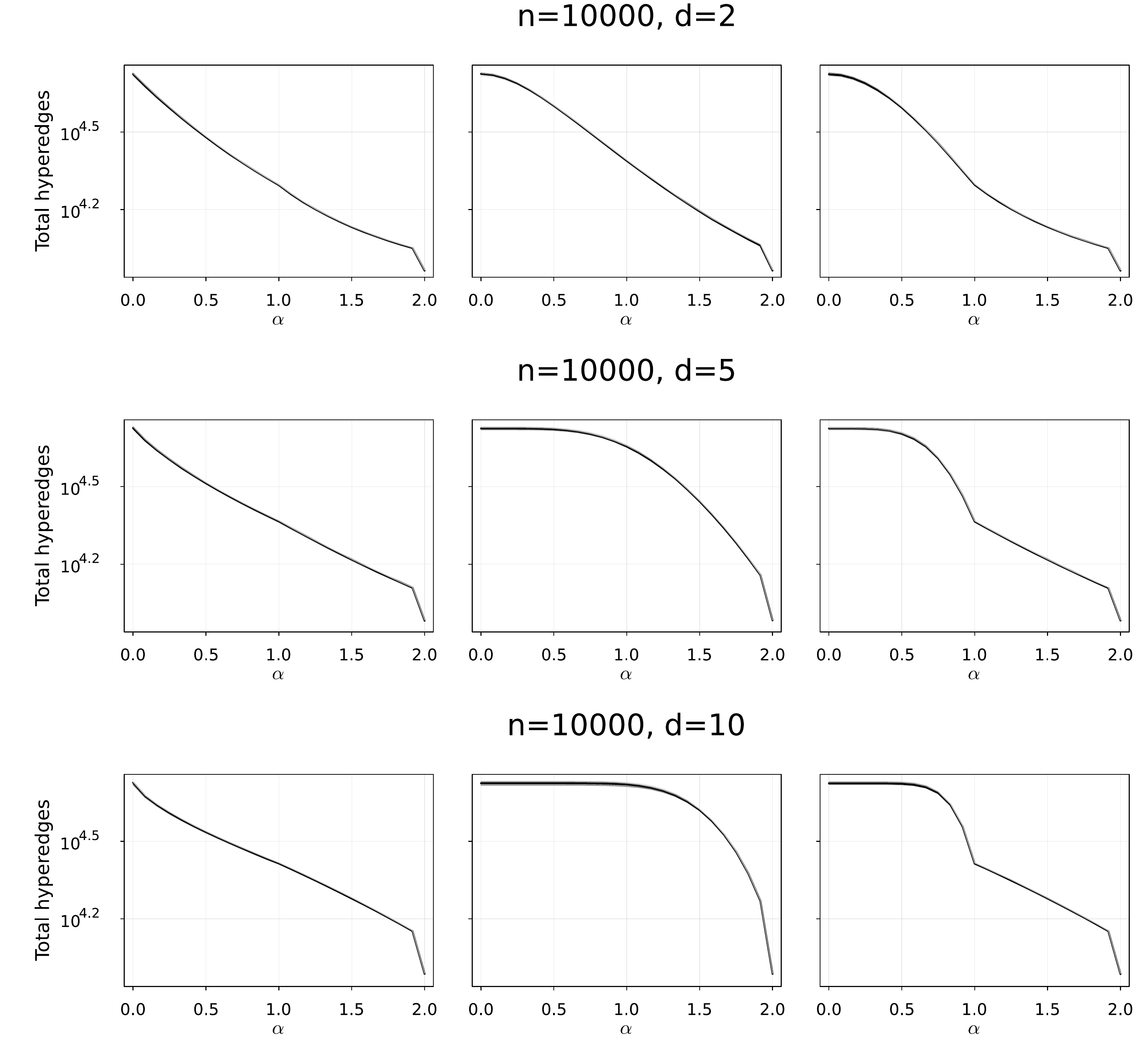}
    \caption{Total number of hyperedges formed using functions for scaling the radius parameter in DBSCAN as we vary dimensions. The functions used for $\epsilon_\alpha$ are: Equation~\eqref{eq:dbscan-radius} (leftmost column), linear scaling with $\epsilon_\alpha=\alpha r_v/2$ (middle column), and Equation~\eqref{eq:dbscan-radius-alt} (rightmost column). Gray to black to gray bands (only visible when zoomed in) in the middle and right columns indicate 10th, 25th, 50th, 75th, and 90th percentiles. This shows that our choice of Equation~\eqref{eq:dbscan-radius} gives the smoothest interpolation between pairwise effects and pure hyperedge effects as $\alpha$ varies from $0$ to $2$.}
    \label{fig:different-alpha-choices}
\end{figure}

\subsection{Invariance of connected components}
We next show that the connectivity of the overall graph or hypergraph is invariant to the choice of $\alpha$ and depends only on the set of points $\mX$ and the assigned degrees $\vd$. This is intuitively straightforward because the edges formed only depend on the points and the specific number of neighbors chosen by $d_v$ for each vertex, but the following argument makes this intuition rigorous. 

\begin{theorem} \label{thm:connectivity}
$\alpha$-Invariance of Connected Components

Let $G_\alpha=H(\boldsymbol{X},\boldsymbol{d},\alpha)$ denote the spatial graph generation model outlined above with parameters $\boldsymbol{X},\boldsymbol{d},\alpha$. For fixed values of $\boldsymbol{X},\boldsymbol{d}$, let $C_{\boldsymbol{X},\boldsymbol{d}}(\alpha)$ denote the connected components in $G_\alpha$. Then $C_{\boldsymbol{X},\boldsymbol{d}}(\alpha)$ is independent of $\alpha$.
\end{theorem}

\begin{proof}
Let $\boldsymbol{X},\boldsymbol{d}$ be fixed and consider two parameters $\alpha_1$ and $\alpha_2$. 
Let $u,v$ denote the same pair of nodes in $G_{\alpha_1}$ and $G_{\alpha_2}$. 
We show that if one graph admits a $uv$-walk then so does the other graph. 
Without loss of generality, let $u=w_0,w_1,\dots,w_k=v$ denote a $uv$-walk in $G_{\alpha_1}$. 
For any two adjacent nodes on this path, $w_iw_{i+1}$, there exists some hyperedge $h_i\in E(G_{\alpha_1})$ where $w_i,w_{i+1}\in h_i$. 
However, $h_i$ has an associated node used to form this hyperedge during the graph generation process, call it $z_i$. 
Then $w_i,w_{i+1}$ are contained in the $\boldsymbol{d}_{z_i}$ nearest neighbors of $z_i$. 
Since $\boldsymbol{X}$ and $\boldsymbol{d}$ are fixed, this the true for both $G_{\alpha_1}$ and $G_{\alpha_2}$.
Hence $w_i,w_{i+1}$ are adjacent to $z_i$ in both $G_{\alpha_1}$ and $G_{\alpha_2}$. Thus, the sequence $u=w_0,z_0,w_1,z_1,\dots,z_{k-1},w_k=v$ is a $uv$-walk in both $G_1,G_2$.
Hence $C_{\boldsymbol{x},\boldsymbol{d}}(\alpha)$ is independent of $\alpha$.
\end{proof}

A key impact of this result is that it allows us to inherit the usual results about connectedness such as critical thresholds and giant components when $\boldsymbol{X}$ and $\boldsymbol{d}$ are constructed to match such statements. For instance, ref.~\cite{diaz2007dynamicrandomgeometricgraphs} characterizes the behavior of a random geometric graph model in terms of the behavior around a critical connectivity threshold. Because the connectivity of the model is equivalent to the underlying pairwise graph, these same thresholds apply to our model as well. 

\subsection{Graph statistics}
\label{sec:model-statistics}
We continue by empirically studying simple graph statistics. We study what  happens as we vary the dimension $d$ and value of $\alpha$ in Figure~\ref{fig:hypergraph-stats}. We report the average number of hyperedges of a given size, the total number of hyperedges, and the total number of triangles in the pairwise projected graph as a function of $\alpha$. Recall that the pairwise projected graph, or \emph{clique expansion}, results in a clique to represent each hyperedge of the original graph. We use the same experimental setup as in Figure~\ref{fig:different-alpha-choices} and Section~\ref{sec:model-distance}. 

In Figure~\ref{fig:hypergraph-stats}, as we increase $\alpha$ and consequentially $\epsilon_\alpha$, the total number of hyperedges decreases while the total number of triangles in the pairwise projection increases monotonically. This result is expected because the projected graph has more cliques, which adds more triangles. The key point of this figure is that we get larger hyperedges with smaller values of $\alpha$ as the dimension $d$ increases. In terms of the impacts on triangles, this results in a steeper initial increase in triangles, although overall fewer triangles as the hyperedges get larger.

\begin{figure}[t]
    \centering
    \includegraphics[width=0.95\linewidth]{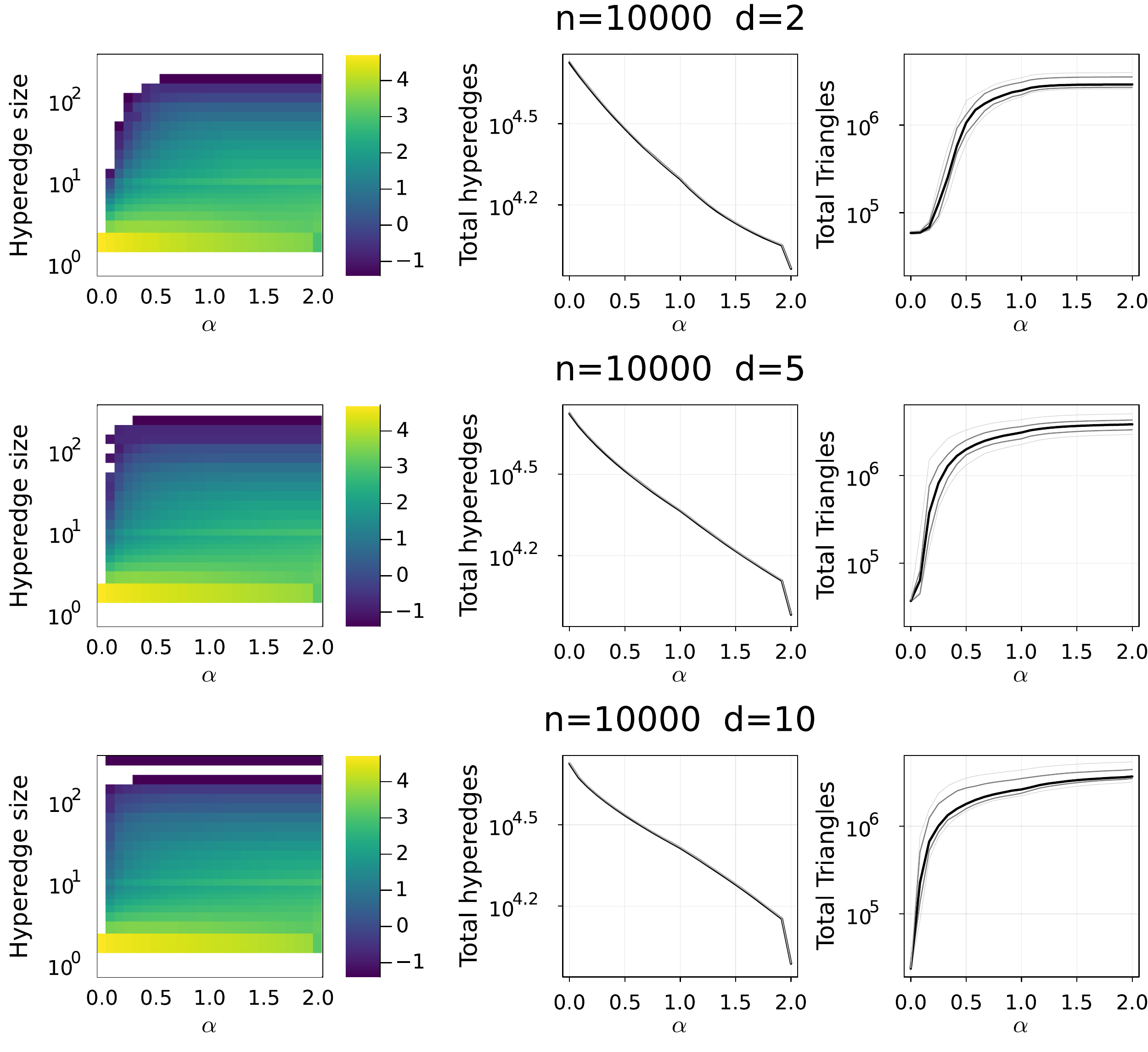}
    \caption{Simple graph statistics as we vary $\alpha$ (Equation~\eqref{eq:dbscan-radius}) and the dimension $d$. The $\alpha$ denotes the interpolation parameter while we show (leftmost column) the average number of hyperedges of a given size, (middle column) the total number of hyperedge present (which is repeated from Figure~\ref{fig:different-alpha-choices}), and (rightmost column) the total number of triangles in the projected graph for $d=2,5,10$. Grey to black to gray bands in the middle and right columns indicate 10th, 25th, 50th, 75th, and 90th percentiles. In the heatmaps, we show the average number of edges in each bin over the 25 trials. The entries are log-scaled, so $4$ corresponds to $10^4$ edges. }
    \label{fig:hypergraph-stats}
\end{figure}

\section{Related Work}
\label{sec:related} 
As mentioned in the introduction, there are a variety of similar geometric graph or spatial hypergraph models, although none of them enable the same type of seamless mapping from pairwise to higher-order that we achieve. In this context, our model extends both pairwise geometric random graph models as well as geometric hypergraph models. We briefly review these constructions and give pointers for more information. We have a longer survey on random hypergraph models in preparation~\cite{Zhu-preprint}. 

\paragraph{Direct inspiration}
Our proposed model draws inspiration directly from the geometric protean model~\cite{bonato2012geometric} and a simplified extension~\cite{bonato2014dimensionality}. These describe a similar latent space model for graphs where nodes have fixed degrees.

\paragraph{Geometric Pairwise Models}
Pairwise random geometric models are created by pairing geometric information $\boldsymbol{X}_v$ for each node $v$ with a distance function $D$ and some rule for connecting nodes that depends on $D$. A common construction is just to connect points below a fixed distance. Given two nodes $u,v$ and their coordinates $\boldsymbol{X}_u, \boldsymbol{X}_v$, an edge is added if $d(\boldsymbol{X}_u, \boldsymbol{X}_v)\le c$ for some threshold $c$. The threshold $c$ may vary with nodes as well. Another common variant is $k$-nearest neighbor (kNN) graphs where each node connects to its $k$ closest nodes. This is akin to using different radii, $r_v$, for each node.
Yet another is to relax to a soft geometric model using a kernel function $f$ by connecting nodes with probability $f(d(\boldsymbol{X}_u, \boldsymbol{X}_v))$ that typically decays as the distance increases. 

See the survey~\cite{barthelemy2022spatial} for an overview of spatial graph models. Spatial models are specific instances of latent space models where edge connection depends on latent node features. A key focus of research on spatial graph models involves connectivity thresholds~\cite{diaz2007dynamicrandomgeometricgraphs}. This has important implications for routing in \emph{ad-hoc} networks of agents where the connection radius is implied by a radio transmission. 

\paragraph{Geometric Hypergraph Models}
There is a completely different notion of a geometric hypergraph described in the survey~\cite{barthelemy2022class}. This alternative model directly creates the bipartite incidence matrix of the hypergraph. The idea is to create two different types of points among the samples of $\mX$. One type of point represents nodes and the other type represents hyperedges. Then we directly build the incidence matrix of the hypergraph by using any of the spatial connection methods for a geometric graph. 
Connectivity properties of a model similar to that in ref.~\cite{barthelemy2022class} was analyzed asymptotically~\cite{de2023connectivity}.

Another class of methods is based on random simplices~\cite{kahle2009topology, Kahle2014, bobrowski2018topology}.
In such methods, topological tools such as the \v{C}ech complex or the Vietoris–Rips complex are used to form hyperedges based on spatial information. Another method~\cite{lunagomez2017geometric} makes use of the same topological tools but places priors on point configurations in order to induce desired structures in the generated simplices.
While of note, these methods are more restrictive than those of hypergraphs. A relaxation from simplices to a geometric hypergraph model of varying sizes was made using latent space modeling and sampling~\cite{turnbull2024latent}.  In particular, the latent space model uses a shared sequence of radii for all nodes $r_2<\cdots<r_K$. A hyperedge of size $k$ is placed among vertices whenever the respective balls with radius $r_k$ intersect.

Our proposed model is distinct from both of these ideas. We allow the hyperedges to vary based on a clustering algorithm instead of a random point selection. Also, we directly generate hypergraphs instead of going through simplicial complexes. 
%\clearpage 

\section{Clustering Coefficients on Hypergraphs}
\label{sec:clustering}

There are two common clustering coefficients for a pairwise graph. 
The global clustering coefficient of a graph $G$ is defined as 
\[C(G) = \frac{6\times T}{P_2},\]
where $T$ denotes the total number of triangles and $P_2$ denotes the total number of length-two paths.
These length-two paths are often called \emph{wedges}.  
The average local clustering coefficient for a graph is given by
\[\bar{C}(G)   = \frac{1}{|V|}\sum_{v\in V} C_v = \frac{1}{|V|}\sum_{v\in V} \frac{T_v}{\binom{k_v}{2}} ,\]
where $V$ is the vertex set, $T_v$ is the number of triangles that contain the node $v$, and $k_v$ is the number of neighbors of node $v$. 

As mentioned in the introduction, when a pairwise concept is generalized to a higher-order structure, it often has multiple generalizations. This is the case for clustering coefficients, and a number of different generalizations of clustering coefficients for hypergraphs have been proposed (see~\cite{miyashita2024clustering,ha2024clustering}). We will compare how a few of these behave on our model as we vary $\alpha$ and the spatial dimension $d$. 

Let $\mathcal{H}$ denote a hypergraph with a set of vertex $V$ and a set of hyperedges $H$. 
Perhaps the simplest such generalization is to project the hypergraph onto a graph via \emph{clique expansion} and then compute the pairwise clustering coefficient for the projected graph. 
In clique expansion, each hyperedge is replaced by a clique over all of the nodes within the hyperedge. We then arrive at two distinct clustering coefficients in the hypergraph based on the global and local clustering coefficients in the projected graph. 

The projected graph is unweighted. However, it is built with a union of cliques. We can also consider the projected multigraph. This is a multigraph interpretation of the weighted projected graph from \texttt{networkx} for example~\cite{Hagberg2008}. In this case, we allow the graph to have multiple edges as we take the union of all the projected cliques. (A weighted version of this same multigraph will arise in our study of epidemics as well.) We use this multigraph to define a weighted local clustering coefficient by counting the number of triangles -- including repetitions due to multiedges -- divided by the number of wedges centered at a node $v$. In this case, it is possible for a node to have \emph{more triangles} than wedges. A scenario where this occurs is if the edge that closes the triangle is repeated whereas the edges defining the wedge are not. For this reason, we clip the maximum value of the local weighted clustering coefficient at $1$. The overall value is then 
\[ \bar{C}_M(G) = \frac{1}{|V|} \sum_v \min \biggl(\frac{T^M_v}{\binom{k_v^M}{2}}, 1\biggr), \] 
here $T_v^M$ is the number of triangles and $k_v^M$ is the number of edge end-points that start at node $v$ -- both including multiplicities due to multiedges. 
We note that there are other notions of a weighted clustering coefficient as well~\cite{Fagiolo2007}. While we are not aware of any place $\bar{C}_M(G)$ has been proposed, we suspect that it has been. 

Another intuitive idea is that triangles are the shortest cycles without repeated edges. This leads to a different generalization of clustering coefficients to bipartite graphs~\cite{robins2004small}. We can apply this idea in the bipartite representation (star expansion) of a hypergraph. 
This clustering coefficient amounts to computing the quantity
\[ CC_4 = \frac{4\times C_4}{L_3},\]
where $C_4$ denotes the number of 4-cycles and $L_3$ denotes the number of 3-paths in a bipartite graph.

In summary, we will study the four quantities: 
\[ 
\begin{array}{@{}rl@{}}
C & \text{global clustering coefficient in projected graph (clique expansion)} \\
\bar{C} & \text{average local clustering coefficient in projected graph (clique expansion)} \\
\bar{C}_M & \text{multigraph local } \\
CC_4 & \text{the bipartite clustering coeff.~in the hypergraph incidence matrix (star expansion)} 
\end{array}\]

\subsection{Experimental results}
We study the clustering coefficients in the same experimental regime from Sections~\ref{sec:model-distance}~and~\ref{sec:model-statistics}. The results of computing the four different clustering coefficients are shown in Figure~\ref{fig:cc-figure}. The bands indicate the minimum and maximum values across $25$ trials, while both the embedding $\boldsymbol{X}$ and degrees $\boldsymbol{k}$ are fixed across values of $\alpha$ for each independent trial.

The first thing we note is that there is no regularity in the behavior as a function of $\alpha$. Both the global and local clustering coefficients ($C$ and $\bar{C}$) initially increase before decreasing for large $\alpha$. We find this result puzzling, as we found that the total number of triangles in the projected graph grows with $\alpha$ in previous figures (Figure~\ref{fig:hypergraph-stats}). We explain this finding in Section~\ref{sec:clustering-scenario} next. 

Our next observation is the critical impact of the spatial dimension $d$ on the results. All of the different clustering coefficients show changes in the behavior with regard to this parameter. For instance, $\bar{C}$ is much lower when $d=10$ compared with $d=2$. Consider also the results with $\bar{C}_M$ as a function of $d$. Here, we observe that $\bar{C}_M$ is \emph{smaller} when $\alpha=1$ and $d=2$ whereas $\bar{C}_M$ is larger. Finally, for $CC_4$ we see an overall decrease in clustering as dimension increases.

In comparison with the three clustering coefficients on the pairwise graph, the behavior of the bipartite clustering coefficient $CC_4$ in the star expansion is more closely aligned with our expectations that as $\alpha$ increases one would expect clustering coefficients to increase.

Overall, these results suggest that clustering coefficients in hypergraphs are far more subtle and complex than in pairwise graphs. This supports the idea that there could be a multitude of reasonable generalizations depending on exactly which features are desirable to capture in the generalization. It also suggests that there probably is not going to be a universal clustering coefficient in all scenarios.

\begin{figure}[t]
    \centering
    
    \includegraphics[width=\linewidth]{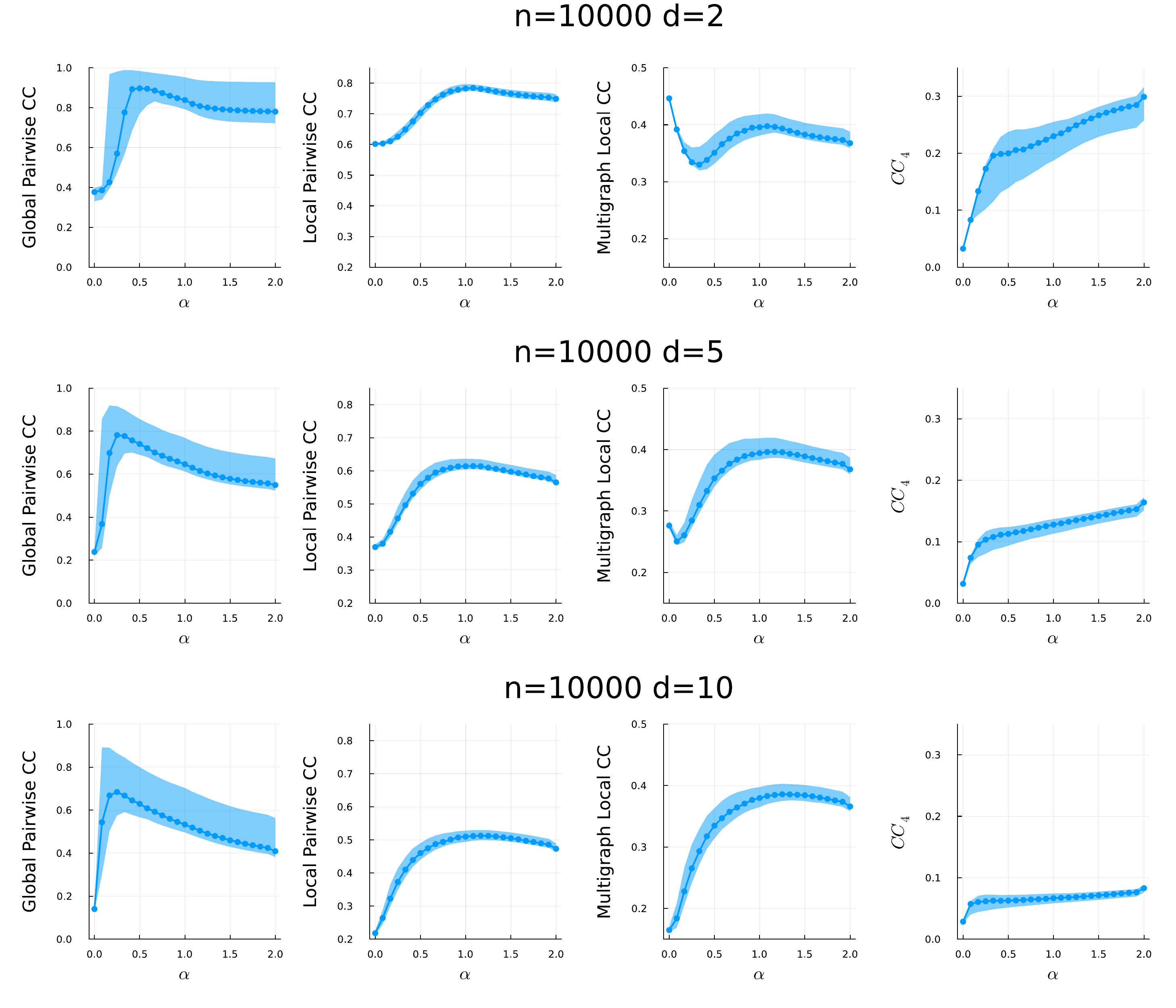}
    \caption{Various clustering coefficients (columns) as we vary $\eps_\alpha$ (x-axes) and the dimension of the node embedding (rows). From left to right, the columns show: (1) pairwise global clustering coefficient, (2) pairwise local clustering coefficient, (3) a multigraph clustering coefficient, and (4) the Robins-Alexander clustering coefficient from the bipartite representation. Bands indicate the maximum and minimum values over 25 trials. Spatial information and degrees are shared across trials.}
    \label{fig:cc-figure}
\end{figure}

\subsection{Why global and local projected clustering coefficients decrease}
\label{sec:clustering-scenario}
We now return to the scenario that initially left us puzzled. Both the local and global clustering coefficient of the pairwise graph $C$ and $\bar{C}$ decrease for large values of $\alpha$. 

Figure~\ref{fig:cc-example} shows a detailed example of how the local clustering of a single node can decrease despite an increase in triangles. 
The left column (panels a,c,e) shows how the nodes $u$ and $v$ are generating hyperedges for $\alpha=1.9$ while the right column (panels b,d,f) shows this process for $\alpha=2$ for our model.
Note the that pairwise edge between $v$ and $u$ in panel (a) is formed when running DBSCAN on the node $v$ as the node $u$ is an outlier or boundary point.

The second row (panels c-d) shows the unweighted pairwise projections of each of these hypergraphs.
Panel (d) has the new pairwise edges relative to panel (c) highlighted in black.
% relative to the smaller value of $\alpha$ highlighted in black.
These new pairwise neighbors add both triangles and wedges to the node $u$.
In the last row (panels e and f), we show $C_u$ and highlight edges in unclosed wedges centered on the node $u$. We get an unclosed wedge for any combination of left vertices and right vertices in panel (f). 
So there are $3\times3=9$ new wedges centered on the node $u$.
Moreover, in this example the clustering coefficients $\bar{C}, C, CC_4$ also decrease. This shows how adding triangles with a clique expansion in the project graph can, ironically, introduce even more wedges in the graph.

\begin{figure}[tp]
    \centering
    \includegraphics[width=0.75\linewidth]{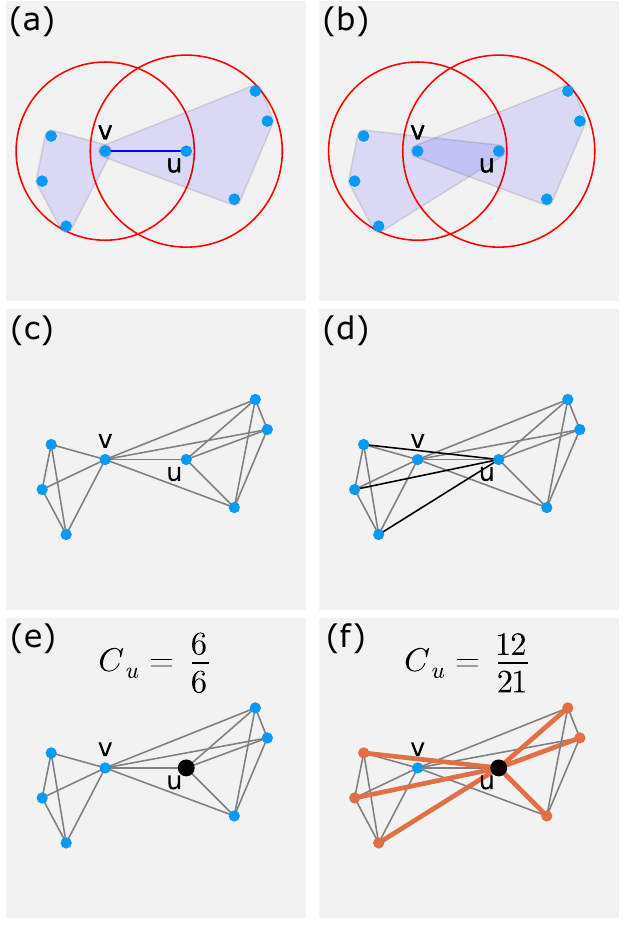}
    \caption{Figure comparing the local clustering coefficient of the node $u$, $C_u$, for two pairwise projections from instances of our model. The first row shows the hyperedges generated by our model. The second row shows the pairwise projections from row 1, highlighting the different edges in black in the right column. The third row highlights all edges that participate in a wedge centered about the node $u$. There are none for the left column but 9 wedges in the right column. $C_u$ is displayed as a ratio of triangles to wedges centered on the node $u$.}
    \label{fig:cc-example}
\end{figure}

%\clearpage 

\section{Diffusion models in spatial hypergraphs do not show higher-order effects}
\label{sec:diffusion}

% parameters/ model details. for internal use
% # hypergraph ppr parameters 
% q = 2.0
% L = LH.loss_type(q) # the loss-type, this is a 2-norm)
% kappa = 0.0001 # value of kappa (sparsity regularization)
% gamma = 1.0 # value of gamma (regularization on seed) 
% rho = 0.5 # value of rho (KKT apprx-val)
% # S = [rand(1:n)] # seed set

% # graph construction parameters 
% nnodes = 500 (n)
% dimension = 2 (d)
% alphas = range(0,2,3)
% degree dist: LogNormal(log(3),1)

\begin{figure}[p]
    \centering
    \begin{subfigure}[b]{0.32\textwidth}
        \centering
        \includegraphics[width=\textwidth]{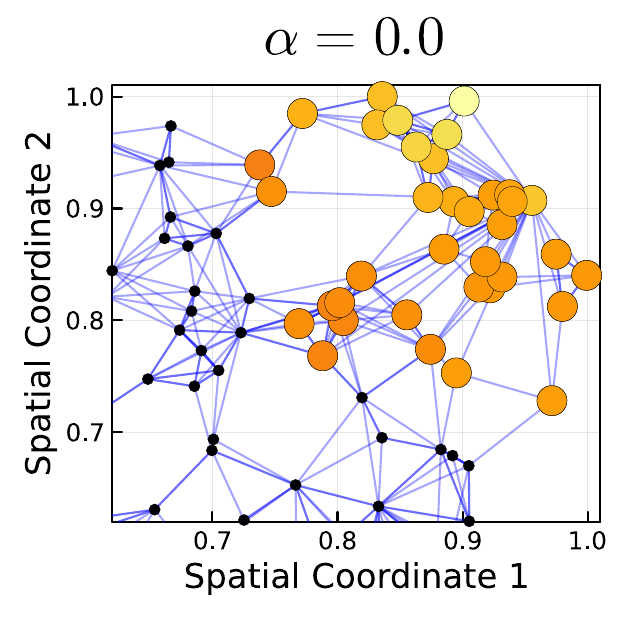}
        \label{fig:ppr-alpha-0}
    \end{subfigure}
    \begin{subfigure}[b]{0.32\textwidth}
        \centering
        \includegraphics[width=\textwidth]{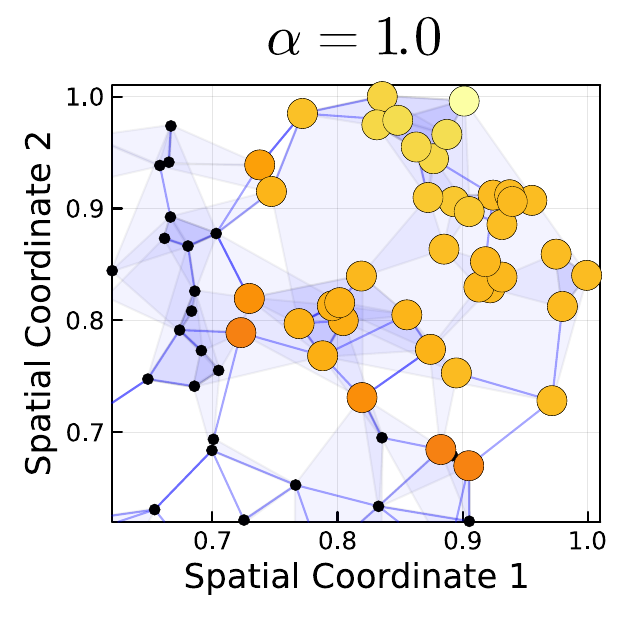}
        \label{fig:ppr-alpha-1}
    \end{subfigure}
    \begin{subfigure}[b]{0.32\textwidth}
        \centering
        \includegraphics[width=\textwidth]{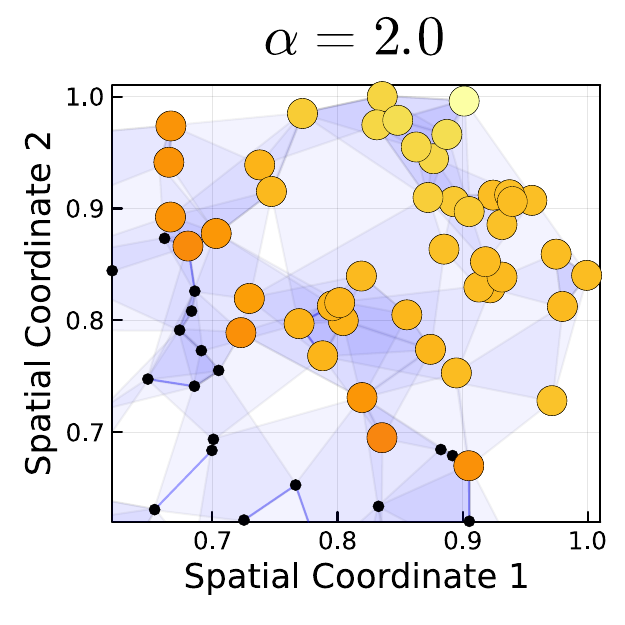}
        \label{fig:ppr-alpha-2}
    \end{subfigure}
    \begin{subfigure}[b]{0.75\textwidth}
        \centering
        \includegraphics[width=\textwidth]{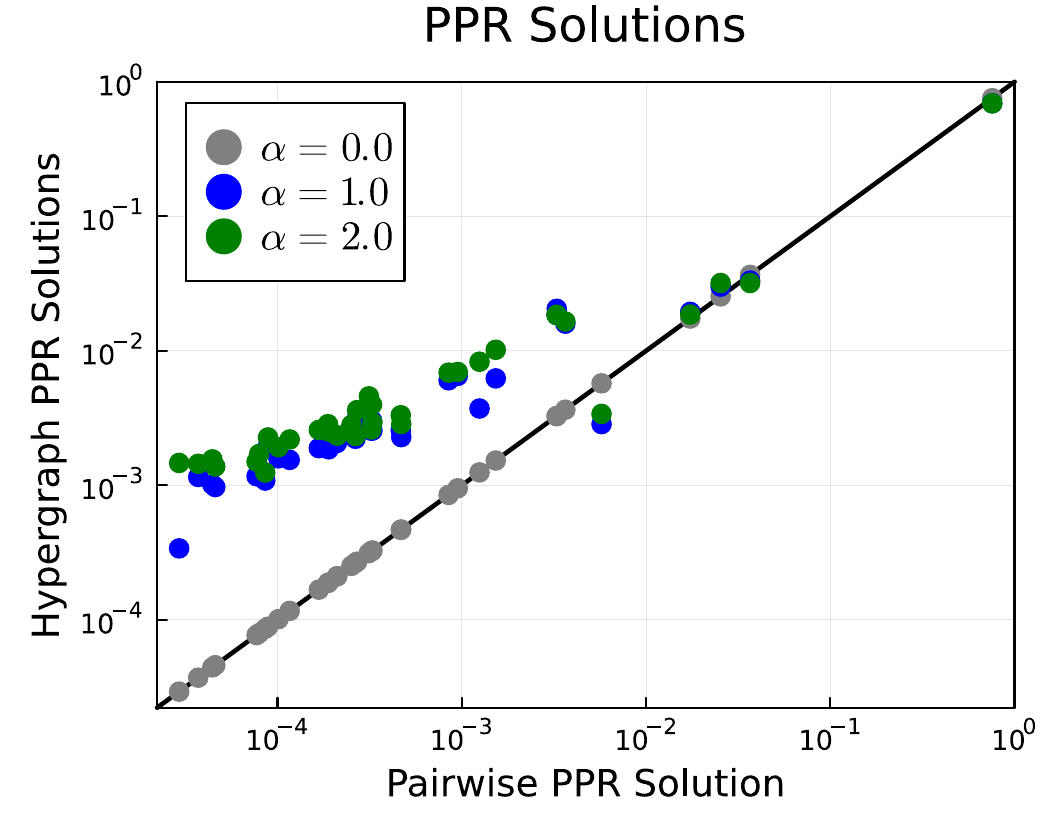}
        \label{fig:ppr-solutions}
    \end{subfigure}
    \caption{Seeded and sparse PageRank solutions in our spatial hypergraph model show only minor differences as we vary the amount of higher-order structure. Since PageRank is a diffusion-based method, this is expected as diffusion on the underlying spatial points dominates the effects. The first row shows when seeded on solution for the top-rightmost node for $\alpha=0,1,2$. The colors are chosen based on a log scale and black nodes did not meet the sparsity criteria for inclusion in the solution. The second plot shows a scatter plot of how the  solutions for $\alpha=1,2$ compare to the pairwise case $\alpha=0$. The black line would indicate exact agreement. Here, we see some small differences, but they largely reflect difference in the precise values computed rather than the relative ordering of behavior. }
    \label{fig:ppr-diffusions}
\end{figure}

Next, we use our model to study diffusion. In comparison to clustering coefficients, we do not expect to see large changes in the behavior of diffusion for our graph model. This is because the same spatial substrate underlies the graph as we vary $\alpha$ to interpolate between pairwise and higher order effects. Consequently, the physics of planar diffusion dominates the choice of higher-order or pairwise model. That said, we still see small differences between the two models in our study.

To perform the study, we generate a random set of $n=500$ points in the $[0,1]$ region and a fixed set of degrees for each node sampled from a log normal distribution with mean $\log(3)$ and variance $1$. (So the mean degree should be around $3.5$). Then we created instances of the graph as we varied $\alpha$. On each graph, we solved a seeded PageRank problem as an instance of a higher-order diffusion. We used the model from~\cite{Liu-2021-localhyper}, although we'd expect similar results from other hypergraph generalizations of PageRank such as~\cite{Li-Hypergraph-Pagerank,Takai2020}. The algorithm we pick is a strongly local and sparsity promoting PageRank method that we choose because it might encourage slightly more differences among the models than global seeded PageRank solutions. The specific parameters we used were a $2$-norm loss, $\kappa=0.0001$ (this causes the solution to grow away from the seed), $\gamma=1$ (this reflects the ``width'' around the seed), and $\rho=0.5$ (an approximation tolerance). Then we pick a seed in a corner of the graph and zoom in on the region identified by the seeded PageRank vector. 

The results when using $\alpha=0,1,2$ are shown in Figure~\ref{fig:ppr-diffusions}. These show that the behavior of the diffusion propagates radially away from the seed vector. There are some small differences in the propagation, especially around the boundary where the solution entries are small and the sparsity truncation changes which entries are truncated slightly. On the other hand, we see broad agreement in terms of element magnitude among the cases. We compare the values more closely in the scatterplot which shows that there are differences in the values, but the relative ordering is largely preserved with the pairwise case. 

Consequently, and as expected, this model shows little difference on a diffusion computation.

%\clearpage 

\section{Pairwise vs Higher-Order Epidemics}
\label{sec:epidemics}

In our final case study, we illustrate the utility of our spatial hypergraph model for understanding epidemics on hypergraphs. 

While higher-order contagion is relatively new compared to pairwise epidemics, there are a number of significant differences between pairwise and higher-order spreading~\cite{iacopini2019simplicial,higham2021epidemics,li2021contagion,landry2020effect}. In particular, there have been a number of conflicting results on the relevance of higher-order effects in spreading. 
On the one hand, epidemic spread is \emph{inherently} a pairwise behavior in which a real or virtual pathogen spreads from one individual to another in an infection event.
On the other hand, pathogens spread via airborne routes have obvious group-relevant interactions~\cite{Lu2020}. 
Theoretical and empirical studies on these findings have been mixed.
As shown by~\cite{higham2021epidemics}, without strong hyperedge-dependent infection effects, hyperedge transmission models reduce to weighted pairwise transmission models.
Studies of human mobility and SARS-CoV-2 showed that super-spreading and the associated group interactions were key routes of transmissions~\cite{chang2021mobility,Dixit2023}. Additional theoretical studies show bistable~\cite{iacopini2019simplicial} parameter regions in epidemics with simplicial complexes.

We use our model to study the impact of group-level or higher-order spreading in epidemics. Our vision is to model an airborne pathogen where group-level behavior may be important. 
We make several simplifying assumptions. 
A single infectious node in a group is enough to infect any other node within that group.
This is in contrast to collective contagion in which some fraction of nodes must be infectious to enable group-level spreading.
Moreover, we scale the probability of a node becoming infected with both the number of infected nodes within that group as well as the size of the group in different ways. This is because the more infectious contacts a node has within a group, the more infectious aerosols would be produced. 
However, a joint contact among a large group requires more space. Moreover, ventilation standards in the US~\cite{ASHRAE2019} state air dilution rates that scale with the number of people. This dilution will be a key feature in our models.

\subsection{The epidemic model} 
We used a discrete time Susceptible-Infected-Recovered-Susceptible (SIRS) compartment model with an additional exogenous infection term.
At each time, $t$, a node can be in exactly one of three states: susceptible, infected, or recovered.
Recovered nodes are temporarily immune from infection and they lose immunity with probability $\delta$. Infected nodes transition to recovered with probability $\gamma$. 
Susceptible nodes become infected due to contacts with infectious individuals through an edge or hyperedge or an exogenous infection with probability $\eta$. We include this exogenous infection term because we model a small population embedded within a larger population that can drive infections through other means.  In terms of the interaction with the hypergraph, we view each hyperedge as a separate interaction the node experiences during a time period. Consequently, each edge or hyperedge represents a possible infection route. Let $v\in V$ and let $h\in H$ be a hyperedge containing $v$ that represents a group interaction. 
We set the infection probability for node $v$ from hyperedge $h$ at time $t$ to be 
\begin{equation}
    \beta_h^{(t)} = \text{Pr}(v \text{ infected by } h | v \text{ susceptible}) = \frac{1-(1-\beta)^{i_{h}^{(t)}}}{g(|h|)},
    \label{eq:hyper-beta}
\end{equation}
where $\beta$ denotes a baseline pairwise infection probability, $i_{h}^{(t)}$ denotes the total number of infected nodes in the hyperedge $h$ at time $t$, and $g$ is function that represents the impact of ventilation that depends on the total number of nodes in $h$.
The term $1-(1-\beta)^{i_{h}^{(t)}}$ represents the pairwise infection probability within the hyperedge $h$ for a susceptible node. 
We further assume independence among distinct hyperedges.
So two nodes can interact among several groups and each of those interactions can independently transmit infection.

We make three simple choices for $g(m)$ to represent various potential ventilation scenarios. We use $g(m)=1,\sqrt{m},m$.
The case of $g(m)=1$ corresponds to no ventilation for group interactions and hence infections are transmitted in a pairwise fashion within each hyperedge.
The case of $g(m)=m$ corresponds to a linear dilution due to improved airflow. The US ventilation standards~\cite{ASHRAE2019} provide ventilation rates per person, which should provide increased infectious aerosol dilution in large groups. We also study a low-ventilation scenario where $g(m) = \sqrt{m}$.

\subsection{Related work}

Many methods related to epidemic spreading often makes use of individual-level stochastic models or mean-field approximations of the continuous-time process~\cite{bodo2016sis,suo2018information,higham2021epidemics}.
While these approaches share similarities with discrete event simulations, there are some notable differences in the pairwise case regarding fine-scaled information and the impact of homogeneity assumptions on total infections~\cite{bansal2007individual,volz2011effects, grossmann2021heterogeneity}. 
For this reason, we make use of a discrete event simulator to more accurately model fine-scaled epidemic behavior. 
Thus where other efforts use a rate of infection in continuous time (and the ensuing non-linear term), we directly use probabilities for each discrete time step.  

The biggest difference in our approach is how we treat infections within hyperedges. 
Studies such as ref.~\cite{iacopini2019simplicial} designate simplices as distinct group interactions with special rules for when infection can be transmitted that depends on the number of infected nodes. 
They may also embed all pairwise edges induced by a simplex and treat group spreading as a separate mechanism for spreading from pairwise spread.
In particular, scaling the infection rate with the size of the hyperedge is uncommon. 
For instance, ref.~\cite{bodo2016sis} uses an infection rate that scales with $i_h^{(t)}$ and $\beta$ only but not the size of the hyperedge. 
The only exception we are aware of is ref.~\cite{higham2021epidemics}, which uses a partitioned model that allows a different function for each value of $m$, but their analysis is based on a continuous-time formulation and shows that a mean-field approximation can be reduced to a weighted pairwise model.

\subsection{Results from epidemics on our spatial hypergraph model}

Throughout our remaining experiments, we use $n=50,000$ node graphs in $d=2$ with the same log-normal degree distribution with parameters $\log(3)$, 1, to give a mean initial degree of 3.5. We use the same exogenous infection rate for all experiments, $\eta = 5/1,000,000$, which corresponds to one exogenous infection every 4 time-steps. We also use the same recovery parameter $\gamma = 0.1$ for all experiments, which corresponds to an expected infection time of 10 time-steps. We will vary the infectivity parameter $\beta$ and immunity parameter $\delta$ in the experiments.  We run the simulation for 2000 time steps. 

We simulate epidemics using a discrete event simulation to produce trajectories of the number of infected nodes.  Detailed pseudo-code for the SIRS model on a general hypergraph is outlined in Algorithm~\ref{alg:hyper-sirs}. A few sample epidemic trajectories from a pairwise graph are shown in Figure~\ref{fig:total-infections-explaination} (left) where we vary the number of initially infected nodes. These all show convergence to a steady state over the time history of the epidemic. 

As we report on the results for other epidemic parameters, we found that the average number of infected nodes over the last 1000 timesteps was a reliable quantity. We show a histogram of this value over a number of distinct simulations in Figure~\ref{fig:total-infections-explaination} (right). This shows that the \emph{maximum} difference between any simulation was around 100 infections. 

Next, we wanted to ensure that this was a reliable quantity for a larger range of epidemic parameters. We varied $\alpha$ for our graph construction and $\beta$ over a wide range and computed the largest difference in trailing infected nodes to verify that we did not find any large variance. These results are shown in Figure~\ref{fig:no-bistability} for the three different choices of $g$ we consider. Overall, they show small differences -- around 1\% of the total population. 
In particular, based on~\cite{iacopini2019simplicial}, we were worried about possible bistability regimes. These would have manifested as much large changes to the infections over the trials. (Each bin is a maximum difference over 80 distinct simulations.) In fact, the largest values occur right at the epidemic thresholds as we'll see in the next set of figures. The majority of datapoints are even more stable with differences of around 100 infected nodes or 0.2\% of the total population.

% ------------------ HYPER SIRS PSUEDOCODE --------------------------%
\begin{algorithm}[p]
\caption{Hypergraph SIRS Model}
\begin{algorithmic}[1]
\Function{HyperSIRS}{$H,g(\cdot),\beta,\gamma,\delta,\eta,t^{\text{max}}, I_0$} where $\beta,\gamma,\delta,\eta$ are the infection, recovery, waning immunity, exogenous infection probabilities, respectively, $I_0$ initial infected nodes and $g(\cdot)$ is the function modeling the ventilation scenario.
    \For{$v\in V$}
        \Comment{initialize states}
        \If{$v\in I_0$}
            \State $\boldsymbol{s[v]} \gets \text{Infected}$
        \Else
            \State $\boldsymbol{s[v]} \gets \text{Susceptible}$
        \EndIf
    \EndFor
    \For{$v \in V$}
        \Comment{initialized neighbor lists and with number of infected neighbors}
        \State $N[v] \gets [(h,w_{v,h}): v\in h\in H \text{ and } w_{v,h} = \text{number of infected nodes in h}]$
    \EndFor
    \For{$t=1:t^{\text{max}}$}
        \Comment{main loop}
        \For{$v\in V$}
            \If{$rand()<\eta$}
                \State $\boldsymbol{s[v]}\gets \text{Infected}$
            \ElsIf{$\boldsymbol{s[v]}=\text{Susceptible}$}
                \For{$(h,w_{v,h}) \in N[v]$}
                    \If{$w_{v,h}>0$}
                        \State $\beta_h \gets \frac{1-(1-\beta)^{w_{v,h}}}{g(|h|)}$
                        \If{$rand()<\beta_h$}
                            \State $\boldsymbol{s[v]}\gets \text{Infected}$
                        \EndIf
                    \EndIf
                \EndFor
            \ElsIf{$\boldsymbol{s[v]}=\text{Infected}$}
                \If{$rand()<\gamma$}
                    \State $\boldsymbol{s[v]}\gets \text{Recovered}$
                \EndIf
            \ElsIf{$\boldsymbol{s[v]}=\text{Recovered}$}
                \If{$rand()<\eta$}
                    \State $\boldsymbol{s[v]}\gets \text{Susceptible}$
                \EndIf
            \EndIf
        \EndFor
        \For{$v \in V$}
            \Comment{update weights and record stats}
            \State $N[v] \gets [(h,w_{v,h}): v\in h\in H \text{ and } w_{v,h} = \text{number of infected nodes in h}]$
            
    \EndFor
    \EndFor
\EndFunction
\end{algorithmic}
\label{alg:hyper-sirs}
\end{algorithm}

\begin{figure}[tp]
    \centering
    \begin{subfigure}[h]{0.49\textwidth}
        \centering
        \includegraphics[width=\textwidth]{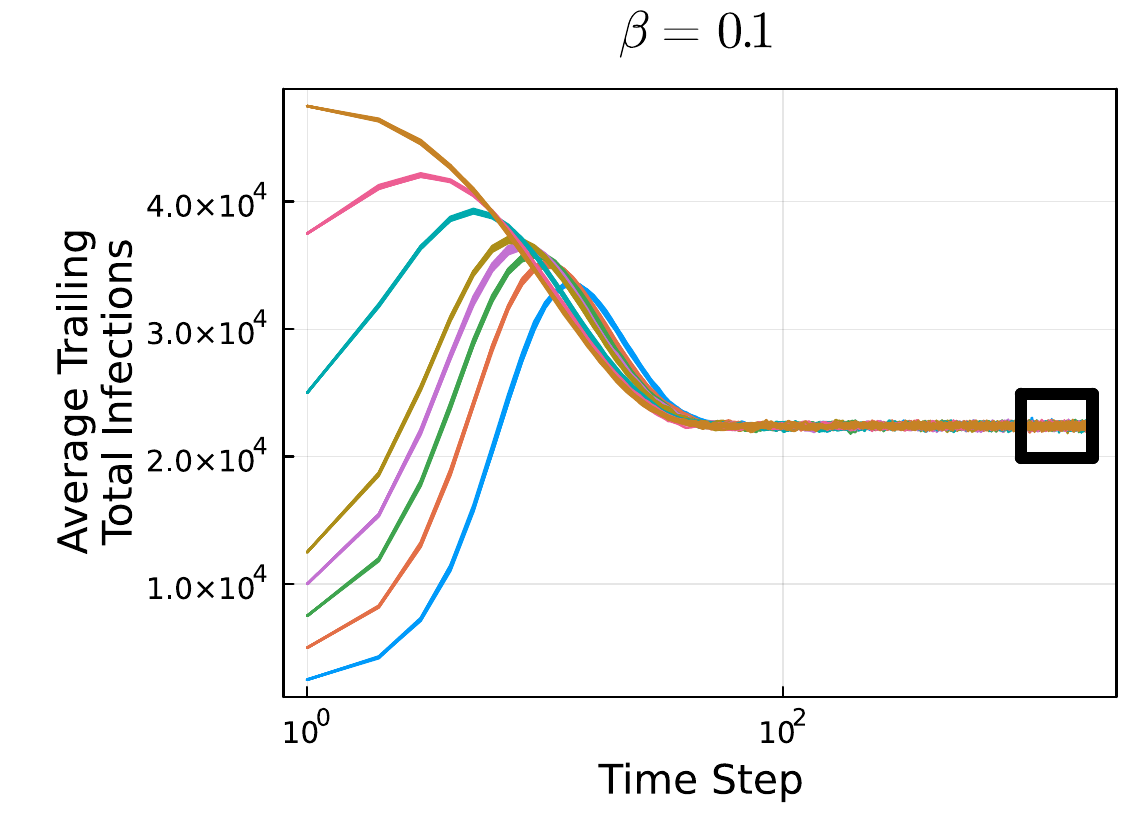}
        \label{fig:total-infections-explaination-1}
    \end{subfigure}
    \begin{subfigure}[h]{0.44\textwidth}
        \centering
        \includegraphics[width=\textwidth]{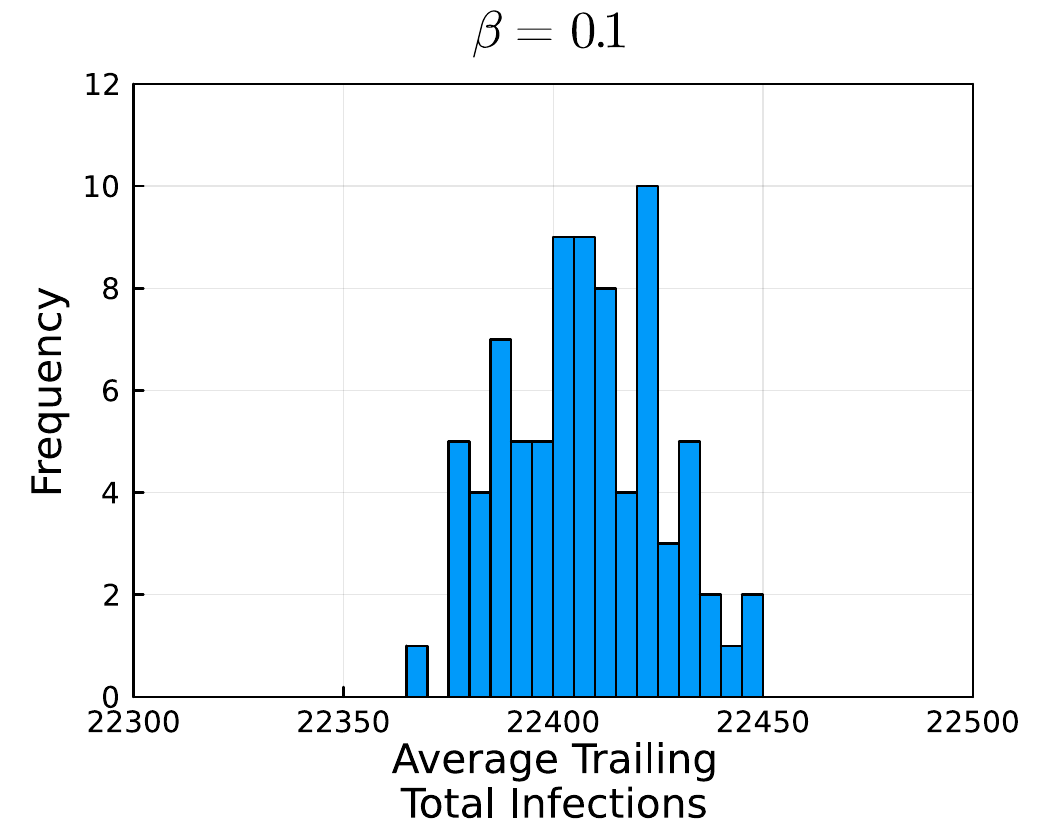}
        \label{fig:total-infections-explaination-2}
    \end{subfigure}
    \caption{(Left) An example of the infected node trajectories on a pairwise graph ($\alpha=0$) with $n=50000$, $d=2$. Despite large changes to the fraction of initially infected nodes, these show convergence to a steady state. The black box indicates the last 1000 times steps, over which we compute the average number of infected nodes. (Right) This shows a histogram over the trailing infected nodes for a number of different samples showing that the average trailing number of infected nodes is a reliable quantity. Other epidemic parameters are fixed at $\gamma=\delta=0.05$.}
    \label{fig:total-infections-explaination}
\end{figure}

\begin{figure}[tp]
    \centering
    \begin{subfigure}[b]{0.3\textwidth}
        \centering
        \includegraphics[width=\textwidth]{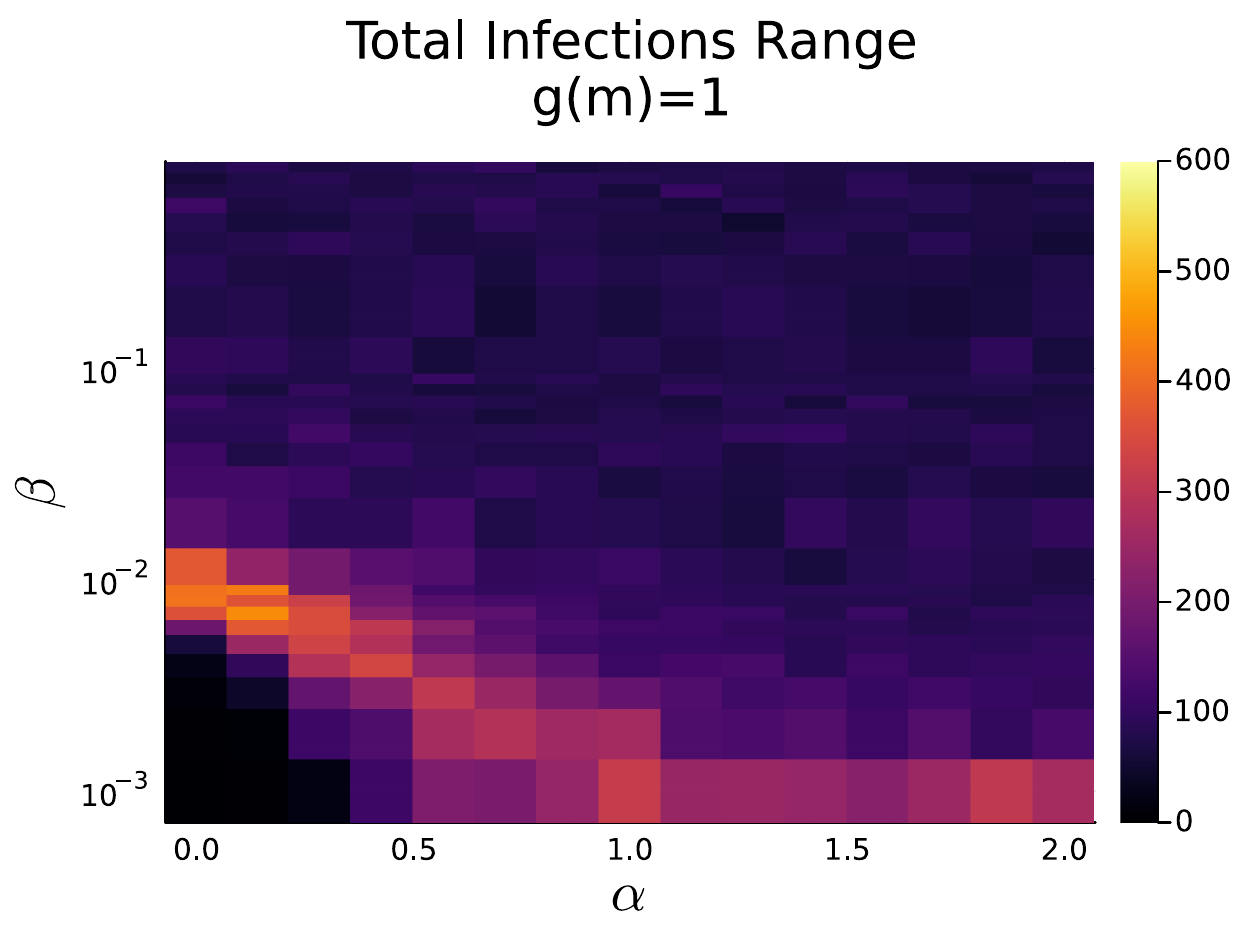}
        % \caption{Total infections for $\beta=4e-1$ and $g(m)=1$}
        \label{fig:no-bistability-1}
    \end{subfigure}
    \begin{subfigure}[b]{0.3\textwidth}
        \centering
        \includegraphics[width=\textwidth]{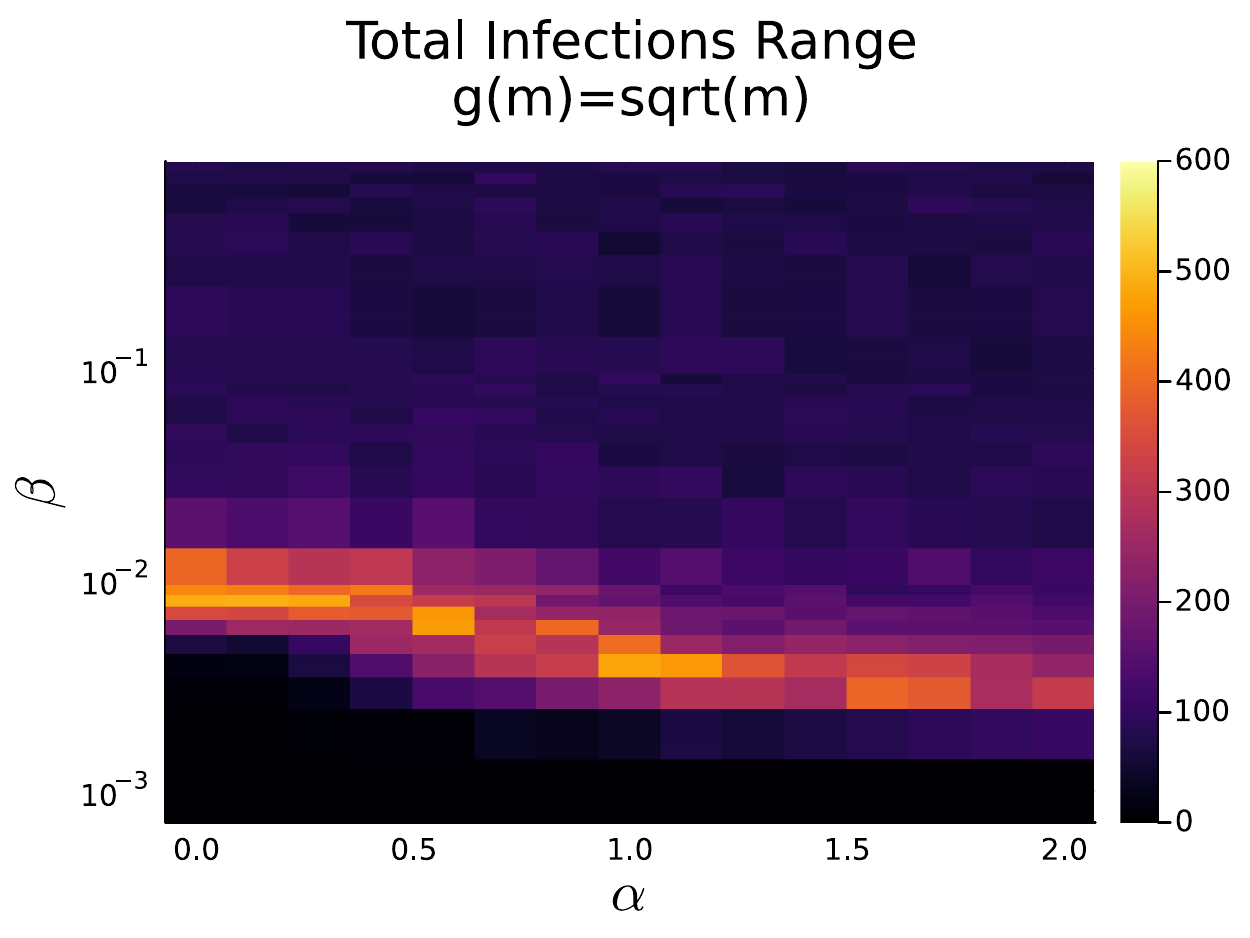}
        % \caption{Total infections for $\beta=4e-1$ and $g(m)=1$}
        \label{fig:no-bistability-2}
    \end{subfigure}
    \begin{subfigure}[b]{0.3\textwidth}
        \centering
        \includegraphics[width=\textwidth]{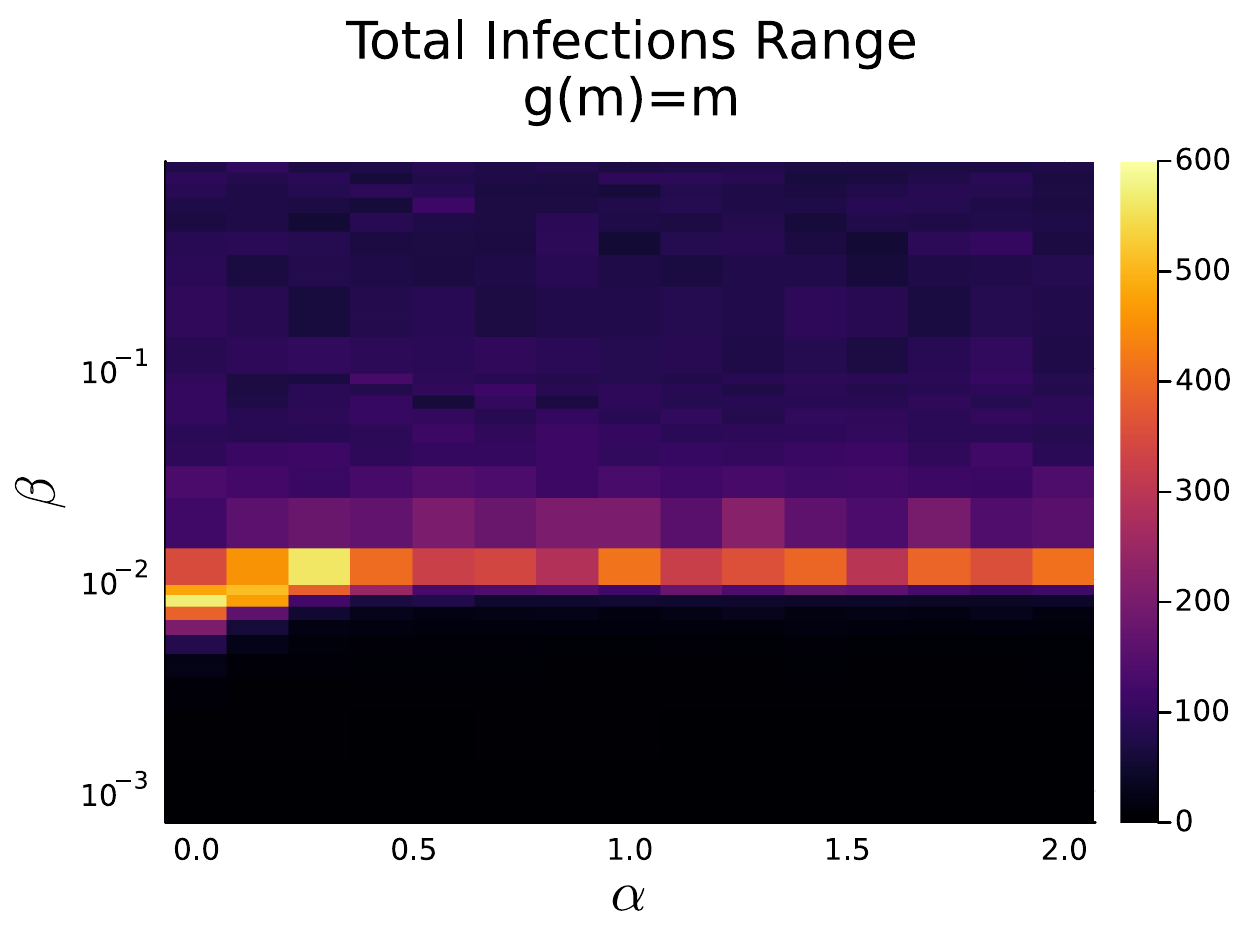}
        % \caption{Total infections for $\beta=4e-1$ and $g(m)=1$}
        \label{fig:no-bistability-3}
    \end{subfigure}
    \caption{The average trailing number of infected nodes is a reliable quantity and we are not in a bistable regime for our epidemic simulations. In this case, epidemics are seeded using an initial number of infected nodes of $5\%, 10\%, 15\%, 20\%,25\%, 50\%, 75\%, $and $95\%$ of the graph. Other epidemic parameters are fixed with $\gamma=\delta=0.05$. Each bin corresponds to 80 simulations (10 per each initial infected fraction). We compute the range of total infections for those simulations in the steady state. The maximum difference is less than 600 infections (which reflects a worse-case fluctuation of about 1\% of nodes) right around the epidemic threshold.}
    \label{fig:no-bistability}
\end{figure}

\begin{figure}[tp]
    \centering
    \begin{subfigure}[b]{0.32\textwidth}
        \centering
        \includegraphics[width=\textwidth]{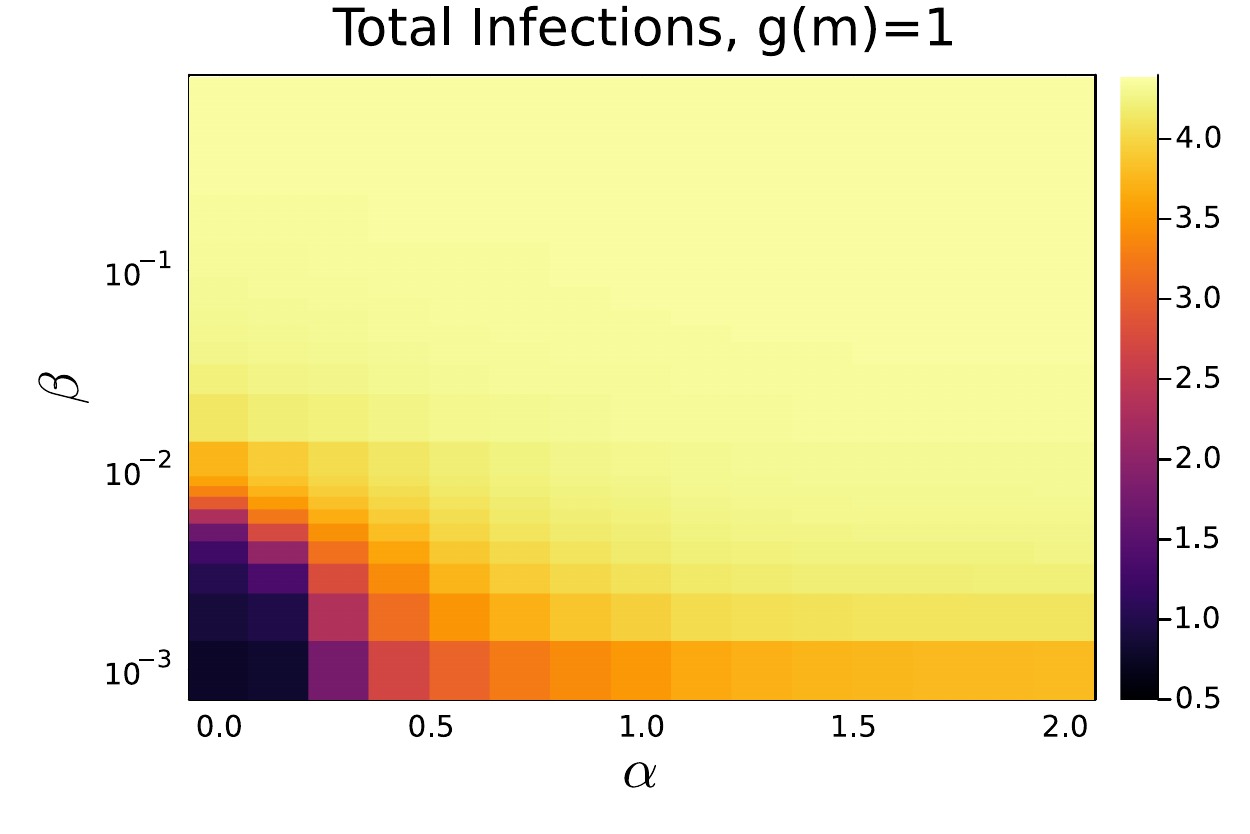}
        \label{fig:total-infections-heatmap-pairwise}
    \end{subfigure}
    \begin{subfigure}[b]{0.32\textwidth}
        \centering
        \includegraphics[width=\textwidth]{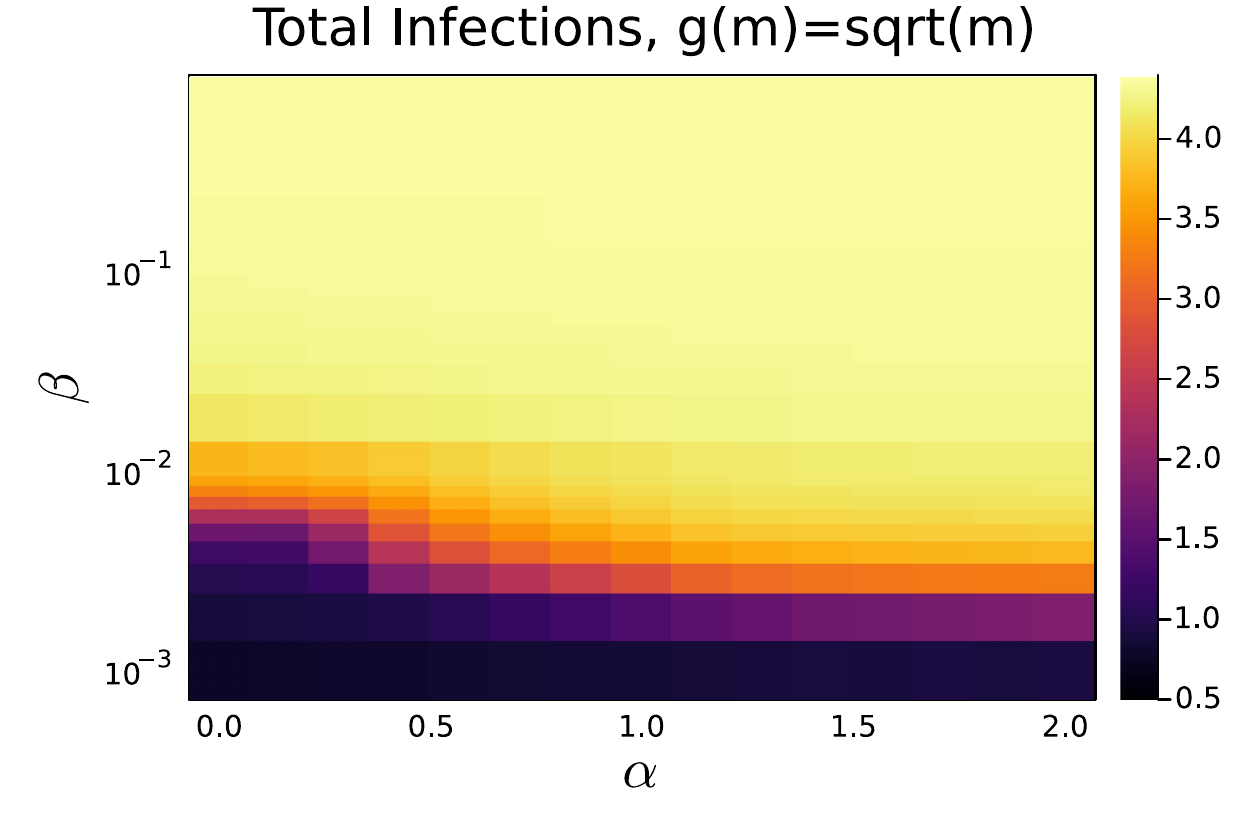}
        \label{fig:total-infections-heatmap-sqrt}
    \end{subfigure}
    \begin{subfigure}[b]{0.32\textwidth}
        \centering
        \includegraphics[width=\textwidth]{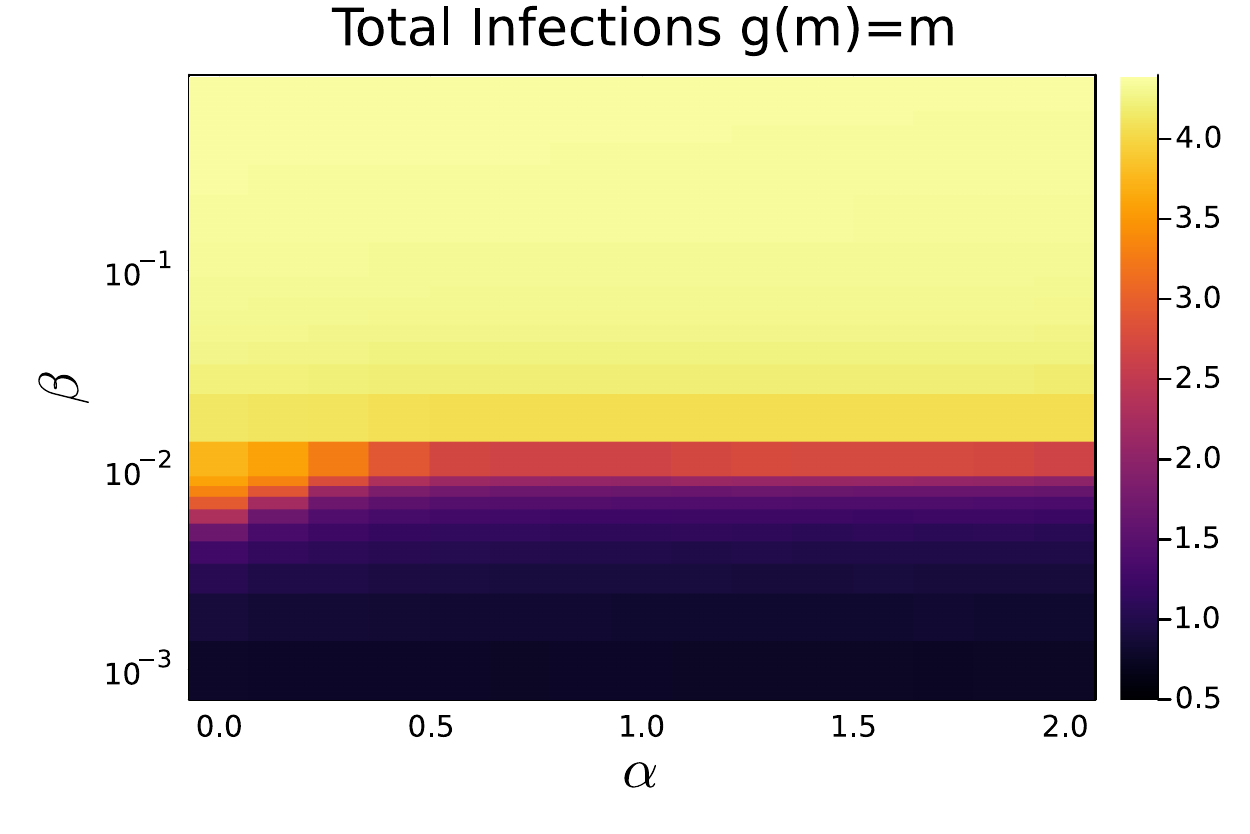}
        \label{fig:total-infections-heatmap-linear}
    \end{subfigure}
    \hfill
    %  NEW ROW 
    \begin{subfigure}[b]{0.32\textwidth}
        \centering
        \includegraphics[width=\textwidth]{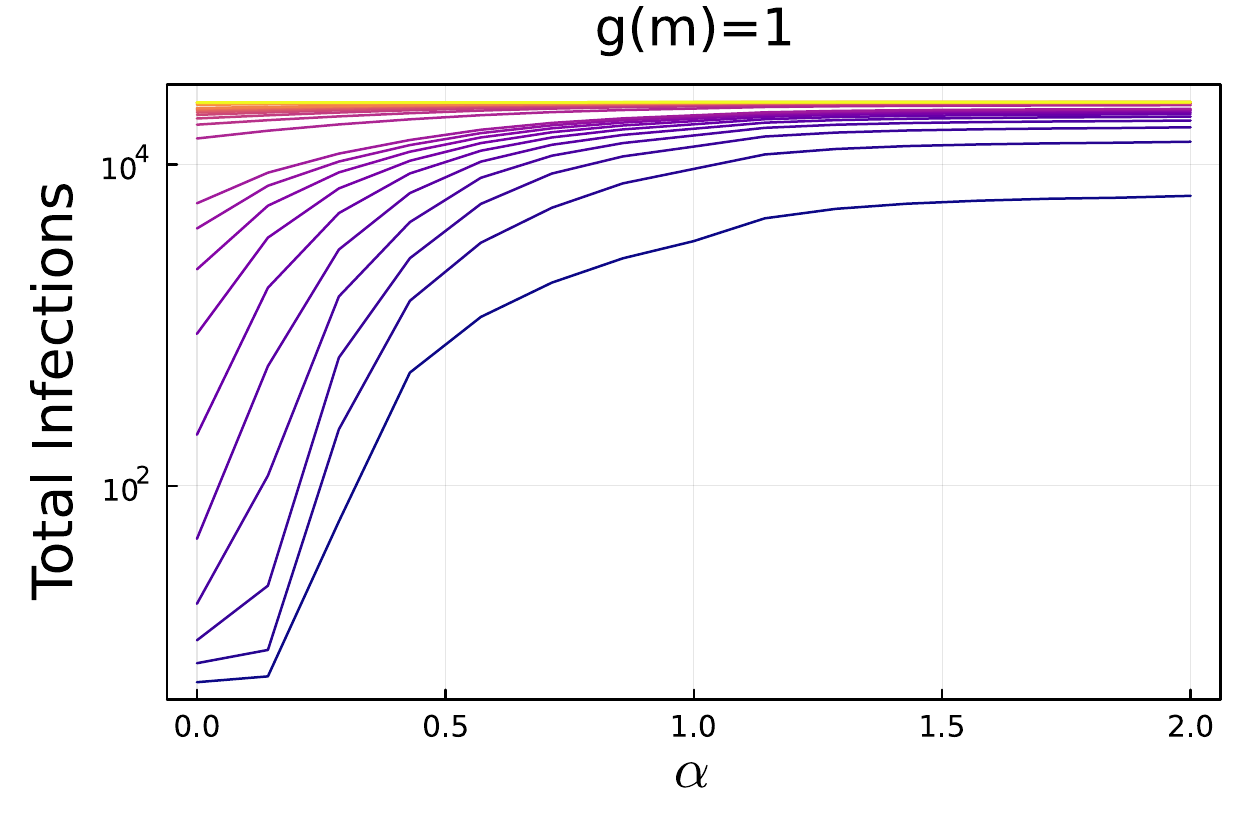}
        \label{fig:total-infections-trajectory-pairwise}
    \end{subfigure}
    \begin{subfigure}[b]{0.32\textwidth}
        \centering
        \includegraphics[width=\textwidth]{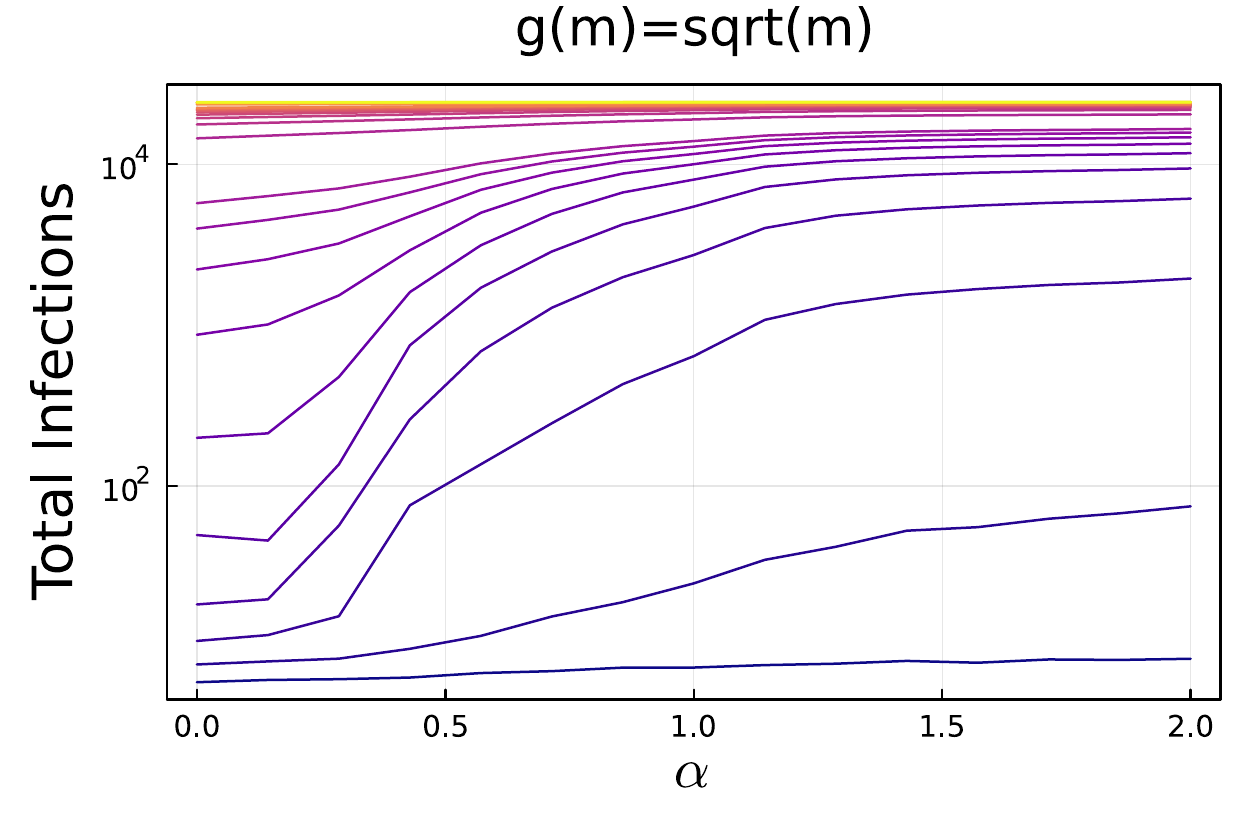}
        \label{fig:total-infections-trajectory-sqrt}
    \end{subfigure}
    \begin{subfigure}[b]{0.32\textwidth}
        \centering
        \includegraphics[width=\textwidth]{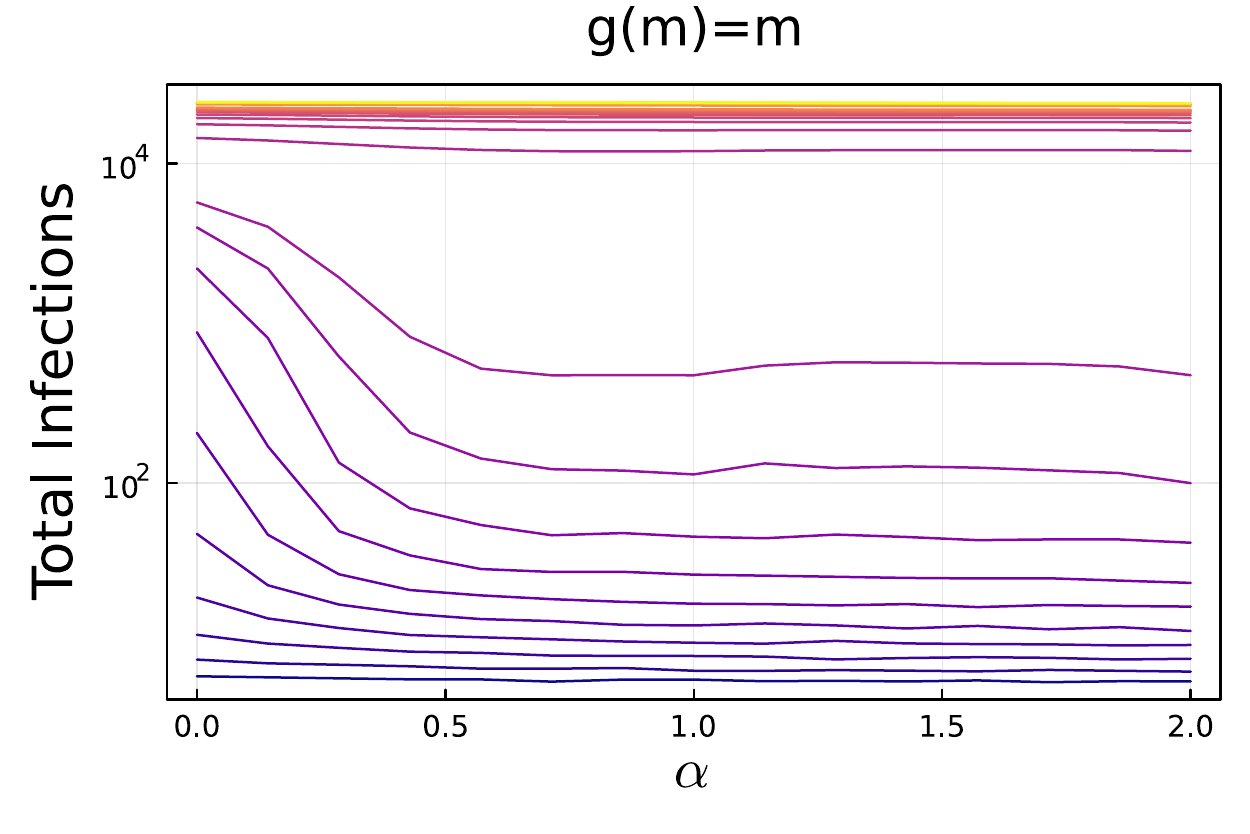}
        \label{fig:total-infections-trajectory-linear}
    \end{subfigure}
    \caption{
    The emulated effects of ventilation and higher-order structure $\alpha$ can have non-intuitive effects on average trailing infected nodes. The top row shows infections as a heatmap in $\alpha$ (higher-order structure) and $\beta$ (infection probability) space while fixing $\gamma=\delta=0.05$. The entries are log-scaled, so 4 means $10^4$ average trailing infected nodes. The bottom row shows the same data as a set of lines where larger values of $\beta$ correspond with more red / yellow colors. The columns correspond to the emulated ventilation functions $g(m) = 1$, $g(m) = \sqrt{m}, g(m) = m$. Note that, in the absence of strong ventilation (left column), increasing $\alpha$ increases total infections. 
    In contrast, under linear ventilation (right column), increasing $\alpha$ decreases total infections 
    }
    \label{fig:total-infections}
\end{figure}

\subsection{Threshold behavior in the epidemics as the infection probability varies}
We next study the same experiment but with a focus on the actual quantity of the average trailing infected nodes. 
Figure~\ref{fig:total-infections} shows the average trailing infections as we vary the amount of higher-order structure in our model ($\alpha$) and the infection probability ($\beta$). 
As expected, each figure shows a clear epidemic threshold. As $\beta$ increases, we go from a small steady state to an outbreak with about half the graph. A second observation is that, as we emulate a greater amount of ventilation ($g(m) = \sqrt{m}$ or $g(m) = m$), we need larger infection probabilities as we transition the graph from pairwise to higher-order structure.

Perhaps the most interesting observation is that the impact of increasing $\alpha$ is coupled with ventilation. 
When $g(m)=1$ then each hyperedge represents a quadratic number of possible infection pathways. This corresponds to increasing the total number of edges, as we will see shortly (Figure~\ref{fig:projected-edge-volume}). Consequently, we \emph{expect} this to show that highly infected populations occur at lower infection probabilities, see e.g., ref~\cite{higham2021epidemics}. 
That this same impact occurs for $g(m) = \sqrt{m}$ is also expected by the same reasoning. This simply doesn't change the probabilities enough to mask the overall increase in effective edges. In contrast, setting $g(m) = m$ should show roughly constant behavior as a function of $\alpha$ by the same reasoning. We do not see this behavior. In the case of linear ventilation (right column) increasing $\alpha$ causes total infections to fall -- especially right around the threshold value of $\beta$. 
This indicates some coupling between ventilation and higher-order structure, which we will explore in the next section. 

\subsection{The impact of hyperedges varies with the epidemic parameters}
\label{sec:varying-impact}
Indeed, the relationship between our ventilation term $g(m)$ and the amount of higher-order structure is not straightforward. In our initial conference paper~\cite{Eldaghar-2024-spatial-hypergraphs}, we found that average trailing infections \emph{increases} with $\alpha$ whereas in the previous section we found that average trailing infections \emph{decreases} with $\alpha$. In a small surprise, the impact is extremely sensitive to the epidemic parameters. 
In the case of linear ventilation, interpolating to higher-order structure can case total infections to increase or decrease in a non-linear fashion.
We illustrate this while separately changing two different epidemic parameters, the infection probability $\beta$ and the waning immunity term $\delta$.
Figure~\ref{fig:hyperedge-effects-p1} shows average trailing total infections as we vary $\beta$ while fixing $\delta=0.01$ under linear ventilation. 
While the number of infections differ among those plots, they show dramatically different shapes depending on how much higher-order structure is present. 
Increasing $\alpha$ can cause total infections to decrease, increase, or produce non-linear mixed effects.
Similar effects are present when varying $\delta$ instead of $\beta$ in Figure~\ref{fig:hyperedge-effects-p2}. 
Note that column 3 of Figure~\ref{fig:hyperedge-effects-p1} is the same as column 1 of Figure~\ref{fig:hyperedge-effects-p2}.

In pairwise epidemics, the dominant eigenvalue is related to the epidemic threshold in the mean-field and often used as a proxy for epidemic strength~\cite{Chakrabarti2008,Prakash2011}.
The total edge volume (total edges in pairwise graphs) is related to epidemic thresholds in randomized networks~\cite{castellano2010thresholds}.
We compute both of these quantities in the pairwise projections to illustrate that the effect we are seeing cannot be explained by simple pairwise tools.
In order to do so, we compute a weighted clique expansion of generated hypergraphs where hyperedges are weighted using the ventilation term $g(m)$.
A single hyperedge of size $m$ is mapped to a clique on $m$ nodes with edge weights $1/g(m)$. We sum up the weights from all hyperedges on the same pair of nodes. 
We then compute the dominant eigenvalue $\lambda_1$ and sum of weighted degrees.
These are the leftmost columns of Figures~\ref{fig:projected-lambda-1} and~\ref{fig:projected-edge-volume} respectively.
In this case, increases in total infections are not due to changes in either the dominant projected eigenvalue of projected edge volumes. 

In terms of mechanisms underlying this effect, we note that the overall changes are modest with respect to the population size. However, they are reliable and repeatable. We were unable to identify any large properties or differences by studying the epidemic propagation in detail, or looking at which hyperedges were responsible, at least beyond similarly small differences. For instance, when we studied the fraction of infection transmissions along edges or hyperedges, we could see that higher-order edges were more likely to transmit infections in the $\delta=0.01$ and $\delta=0.02$ scenario compared with the $\delta=0.03$ scenario. But given that these edges account for a fairly small fraction of the overall transmissions in the highly ventilated case, we did not believe this finding to be mechanistically conclusive.  

\begin{figure}[tp]
    \centering
    \begin{subfigure}[b]{0.24\textwidth}
        \centering
        \includegraphics[width=\textwidth]{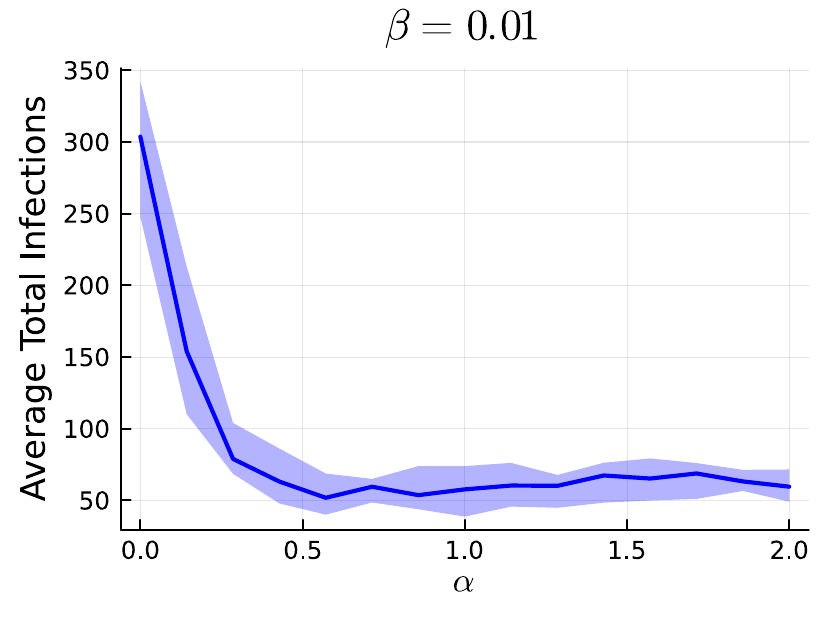}
        \label{fig:hyperedge-effects-p1-1}
    \end{subfigure}
    \begin{subfigure}[b]{0.24\textwidth}
        \centering
        \includegraphics[width=\textwidth]{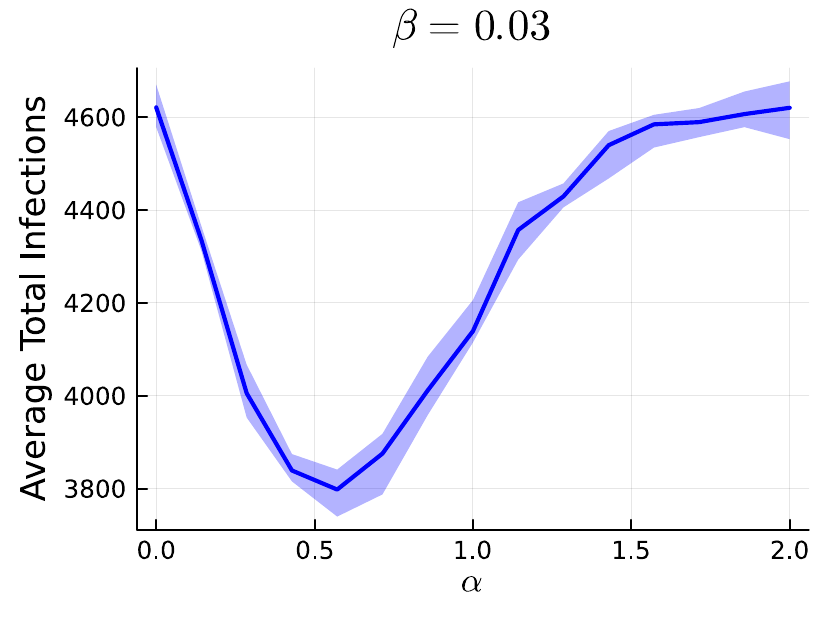}
        \label{fig:hyperedge-effects-p1-2}
    \end{subfigure}
    \begin{subfigure}[b]{0.24\textwidth}
        \centering
        \includegraphics[width=\textwidth]{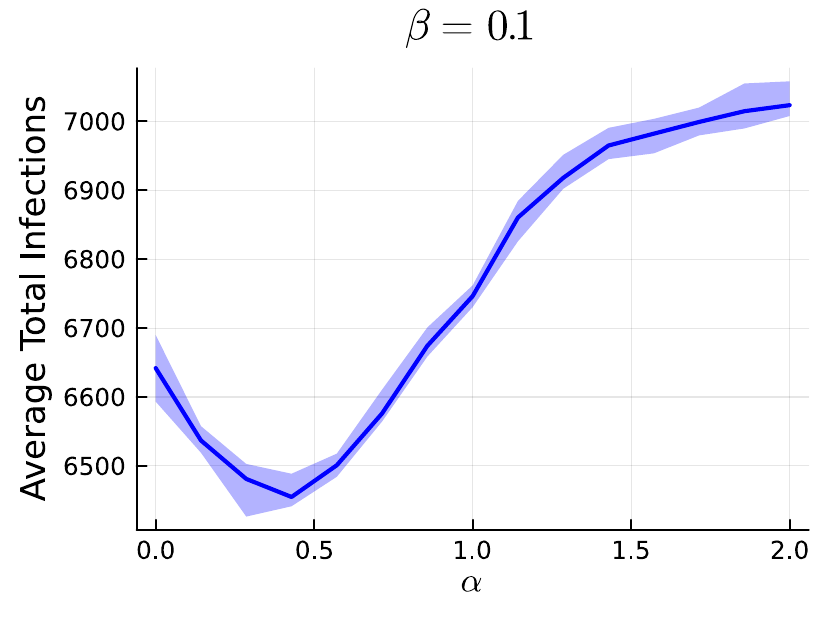}
        \label{fig:hyperedge-effects-p1-3}
    \end{subfigure}
    \begin{subfigure}[b]{0.24\textwidth}
        \centering
        \includegraphics[width=\textwidth]{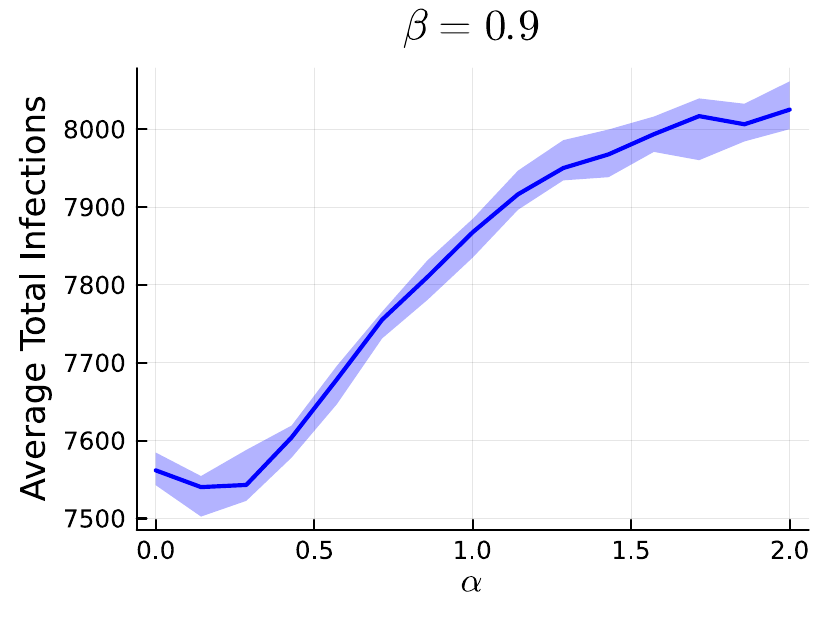}
        \label{fig:hyperedge-effects-p1-4}
    \end{subfigure}
    \caption{The behavior of the average trailing infected nodes metric can vary depending on epidemic parameters. From the left to right, we show what happens as we increase the value of $\beta$ for epidemic simulations for the waning immunity probability $\delta=0.01$ under linear ventilation $g(m) = m$ with $\gamma=0.05$. While average trailing infections differ among the plots, the impact of interpolating from a pairwise to higher-order is sensitive to epidemic parameters.}
    \label{fig:hyperedge-effects-p1}

    \centering
    \begin{subfigure}[b]{0.32\textwidth}
        \centering
        \includegraphics[width=\textwidth]{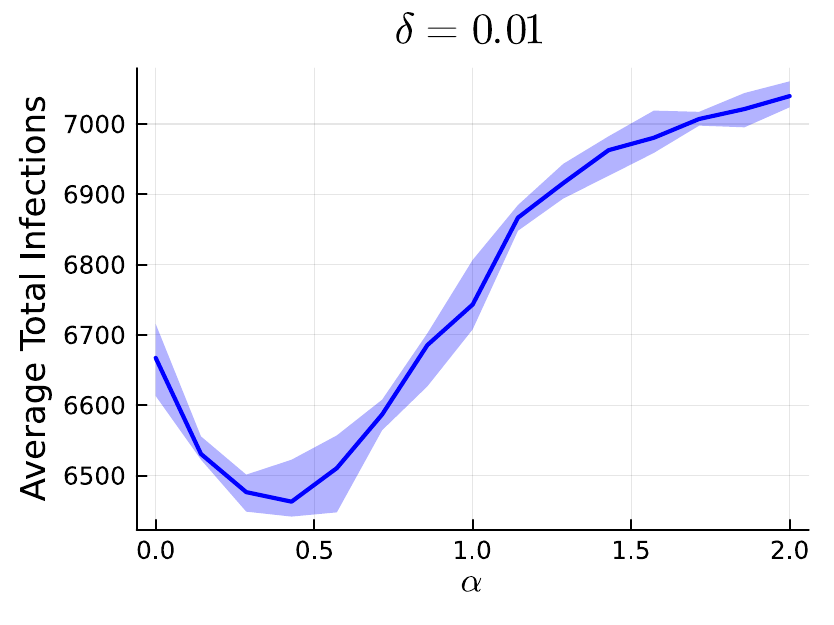}
        \label{fig:hyperedge-effects-p2-1}
    \end{subfigure}
    \begin{subfigure}[b]{0.32\textwidth}
        \centering
        \includegraphics[width=\textwidth]{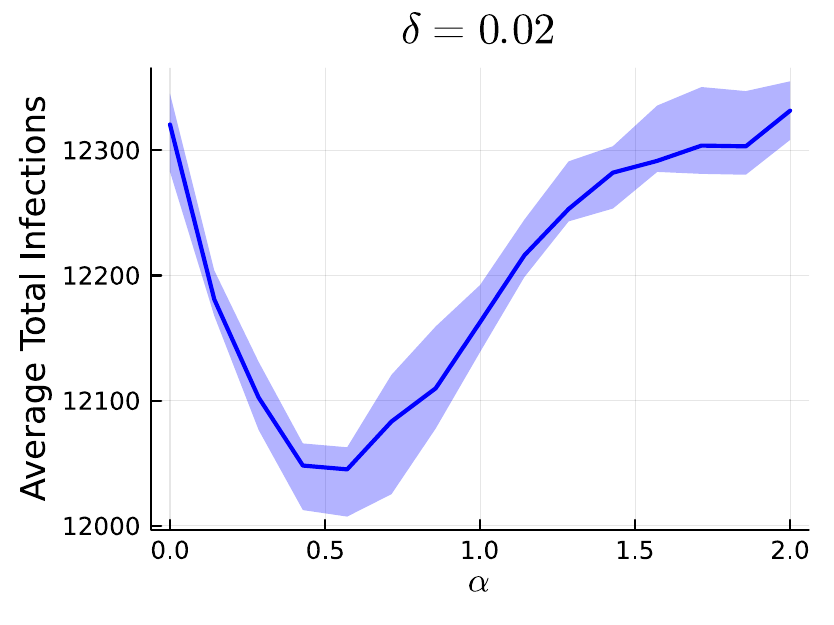}
        \label{fig:hyperedge-effects-p2-2}
    \end{subfigure}
    \begin{subfigure}[b]{0.32\textwidth}
        \centering
        \includegraphics[width=\textwidth]{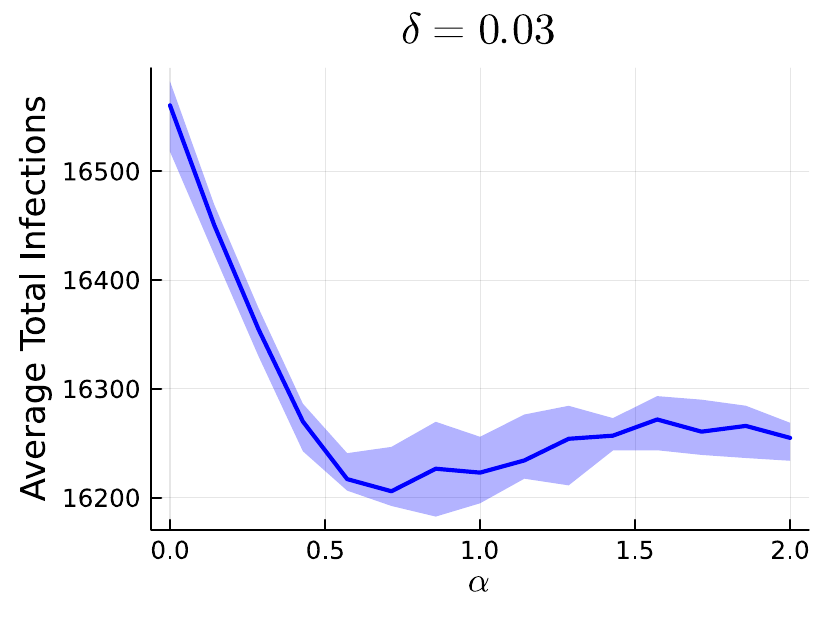}
        \label{fig:hyperedge-effects-p2-3}
    \end{subfigure}
    \caption{Modifying the immune waning term ($\delta$) can also alter the impact of hyperedges under linear ventilation $g(m) = m$. From left to right, we increase the parameter $\delta$ while fixing other epidemic parameters ($\gamma=0.05$ and $\beta=0.1$) and recording average trailing infected nodes. The impact of higher-order structure (large $\alpha$) can cause both growth and decay in the epidemic impact. Note that the leftmost figure is the same as column 3 of Figure~\ref{fig:hyperedge-effects-p1}.}
    \label{fig:hyperedge-effects-p2}

    \centering
    \begin{subfigure}[b]{0.32\textwidth}
        \centering
        \includegraphics[width=\textwidth]{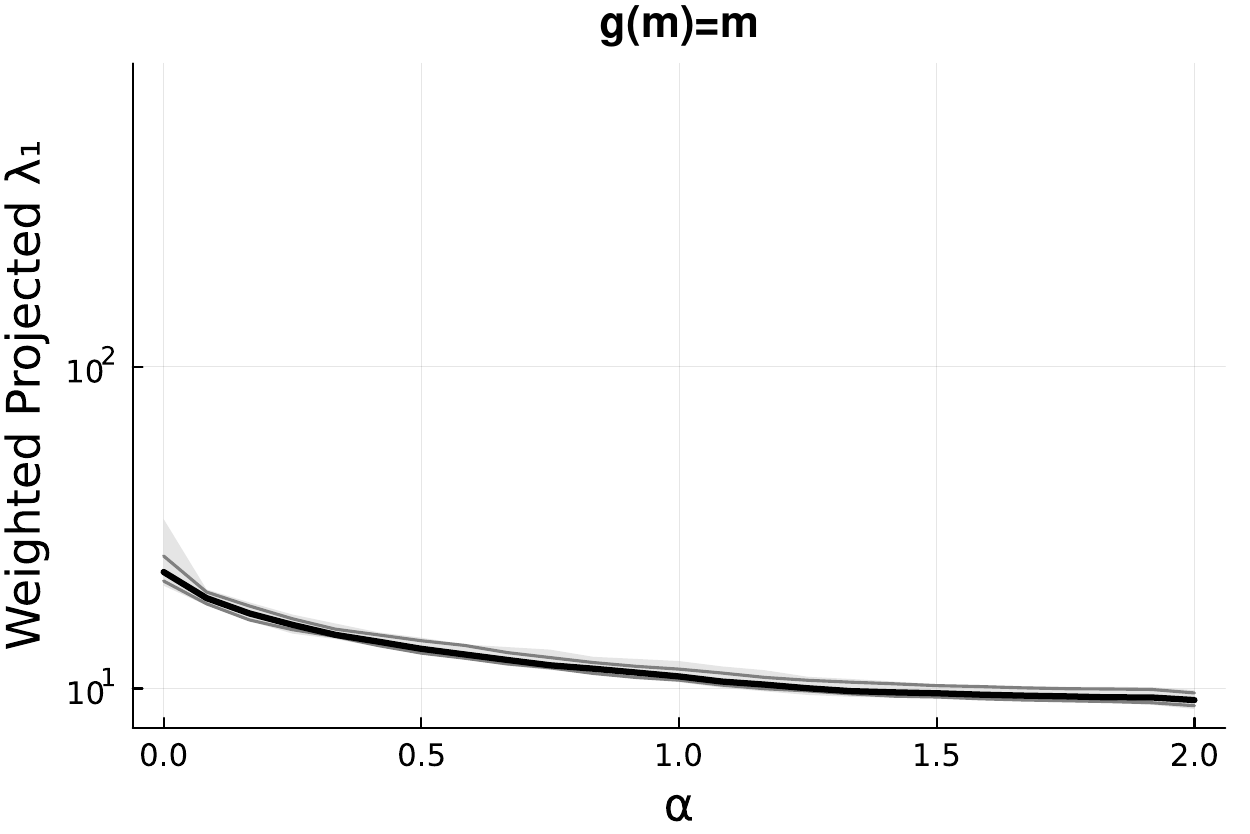}
        \label{fig:projected-lambda1-linear}
    \end{subfigure}
    \begin{subfigure}[b]{0.32\textwidth}
        \centering
        \includegraphics[width=\textwidth]{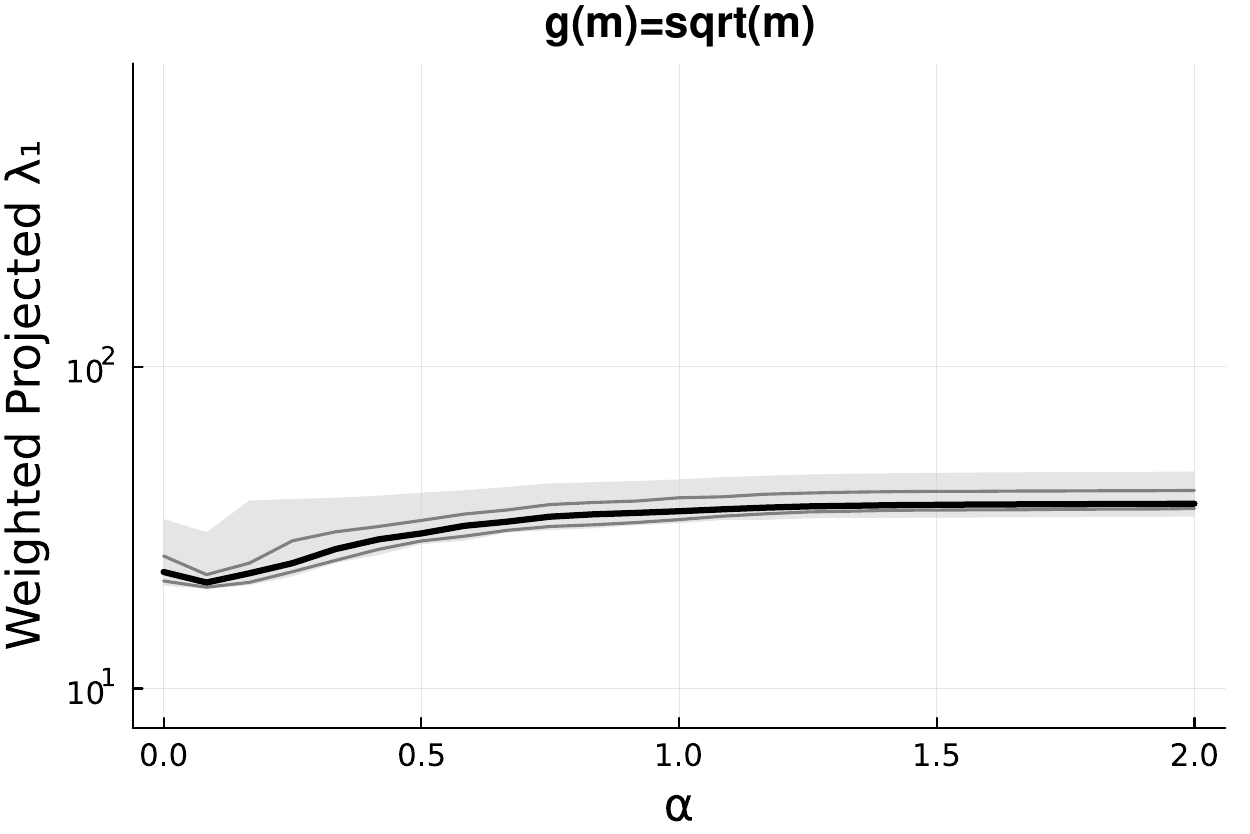}
        \label{fig:projected-lambda1-sqrt}
    \end{subfigure}
    \begin{subfigure}[b]{0.32\textwidth}
        \centering
        \includegraphics[width=\textwidth]{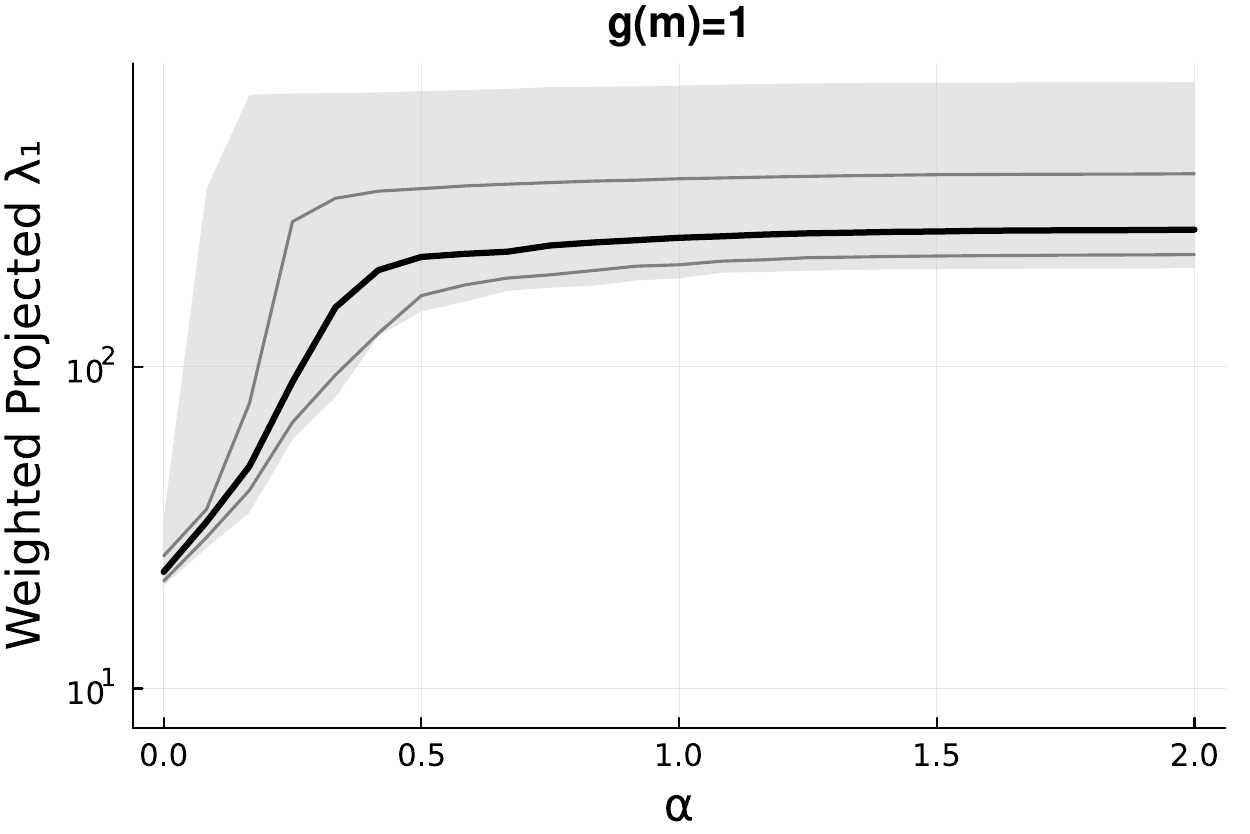}
        \label{fig:projected-lambda1-pairwise}
    \end{subfigure}
    \caption{The dominant eigenvalue of a weighted clique expansion for our spatial hypergraph model while weighing hyperedges according to the ventilation function, $g(m)$. A single hyperedge of size $m$ is projected to a clique with edge weights of $1/g(m)$. 
    Gray and black lines indicate the 10th, 50th, and 90th percentiles while the gray bands denote the max and min over 25 trials.
    Linear ventilation causes the pairwise dominant eigenvalue to decrease while other choices for $g(m)$ generally cause an increase. This would correspond to a weaker epidemic, which is not what is always found in the experiments from Figures~\ref{fig:hyperedge-effects-p1},~\ref{fig:hyperedge-effects-p2}.}
    \label{fig:projected-lambda-1}

    \centering
    \begin{subfigure}[b]{0.32\textwidth}
        \centering
        \includegraphics[width=\textwidth]{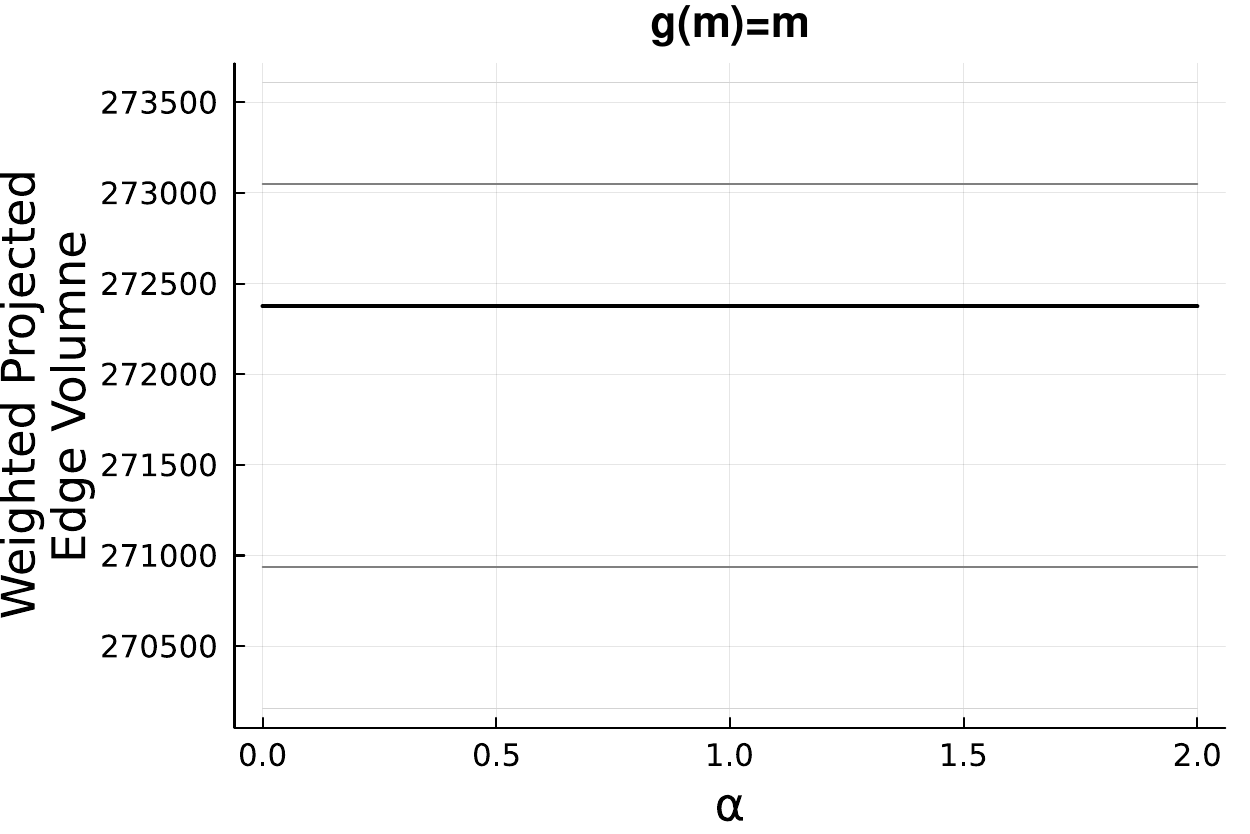}
        \label{fig:projected-edge-volume-linear}
    \end{subfigure}
    \begin{subfigure}[b]{0.32\textwidth}
        \centering
        \includegraphics[width=\textwidth]{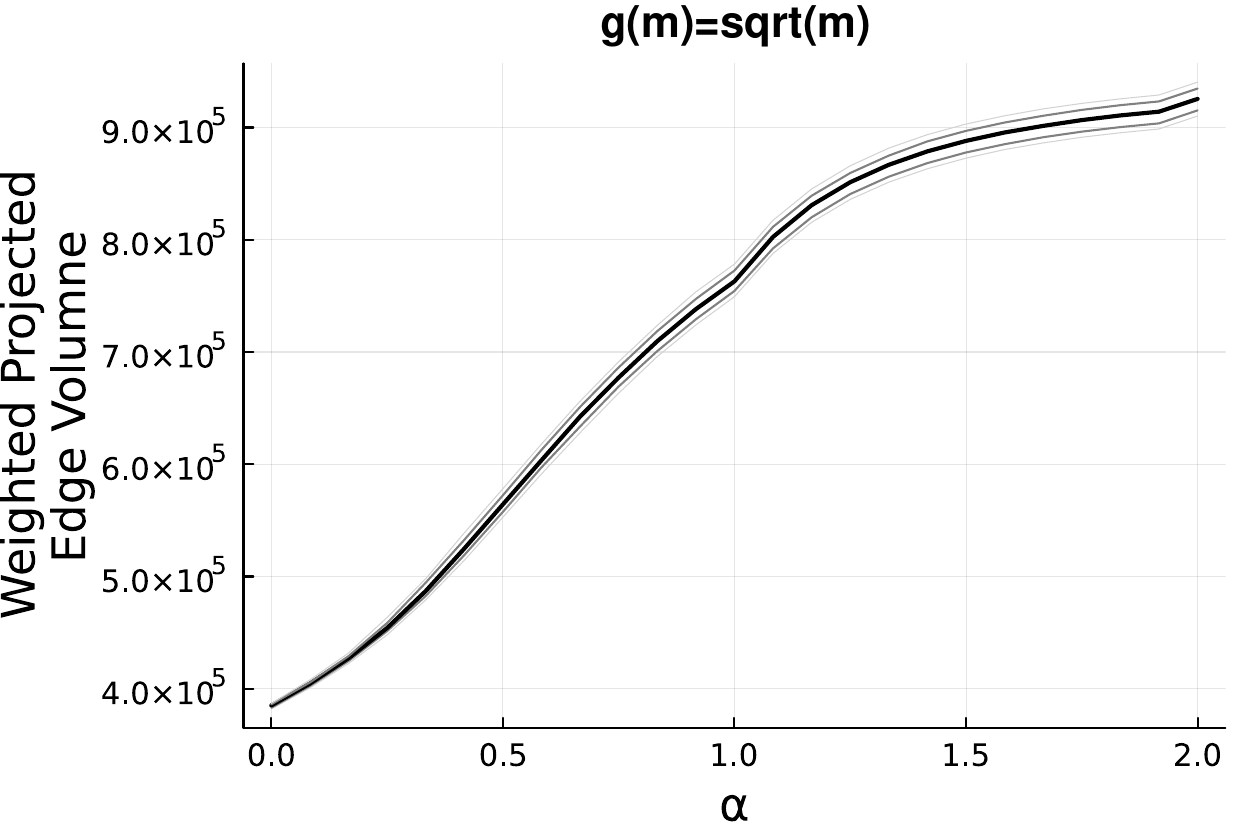}
        \label{fig:projected-edge-volume-sqrt}
    \end{subfigure}
    \begin{subfigure}[b]{0.32\textwidth}
        \centering
        \includegraphics[width=\textwidth]{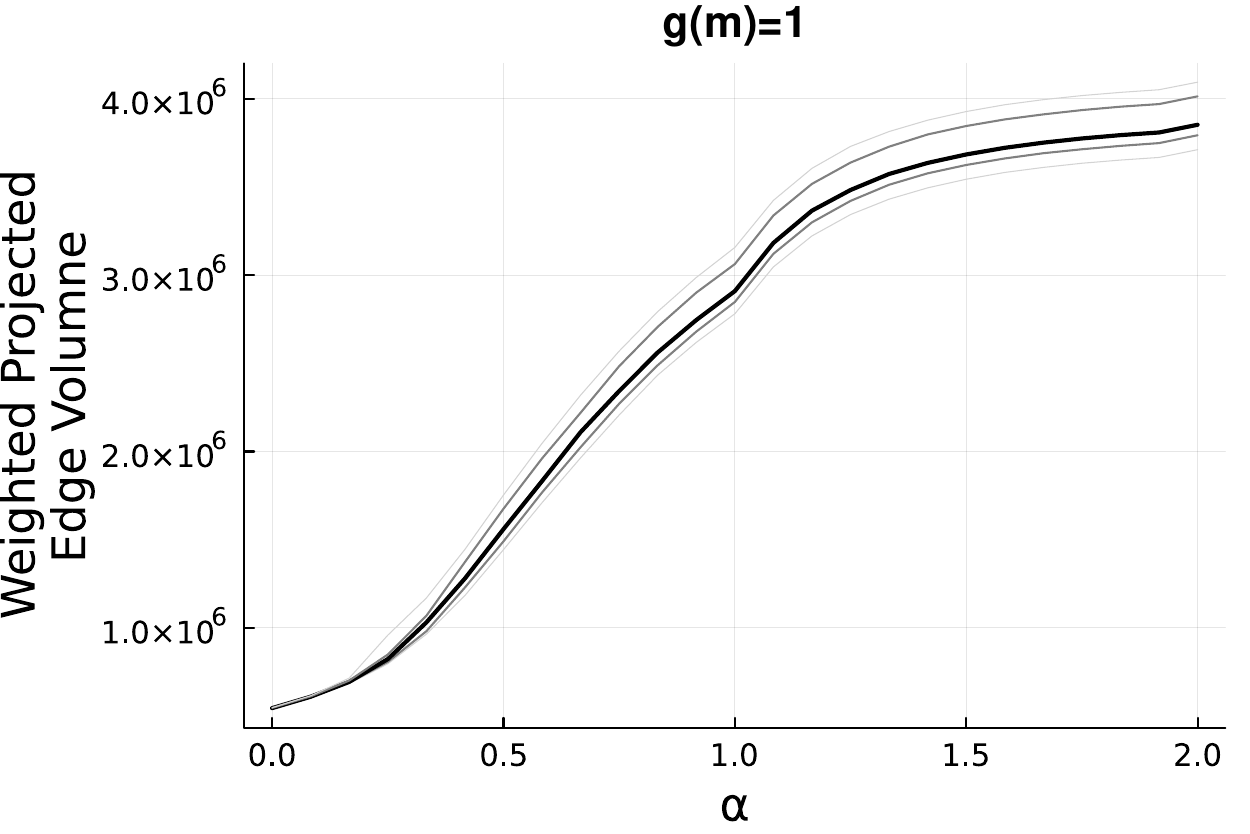}
        \label{fig:projected-edge-volume-pairwise}
    \end{subfigure}
    \caption{The projected edge volume of the weighted clique expansion for our spatial hypergraph model when weighted using $g(m)$. A single hyperedge of size $m$ is projected to a clique with pairwise edge weights $1/g(m)$. Gray and black lines indicate the 10th, 50th, and 90th percentiles.
    This quantity is constant for $g(m)=m$ while increasing for other choices of $g(m)$.}
    \label{fig:projected-edge-volume}
\end{figure}

\subsection{Changes in hyperedge transmissions by size below the pairwise threshold}
Lastly, we examine transmissions by hyperedges of a given size in the steady state. What we call a transmission is when node $i$ infects node $j$ over a hyperedge. 
We focus on the case where there is no ventilation $g(m)=1$.
Recall that this is not a pairwise epidemic as a pair of nodes can interact in multiple hyperedges.
We choose $\beta=10^{-3}$, which is below the threshold in the pairwise case ($\alpha = 0$), and causes a large steady state epidemic in the higher-order case $(\alpha \to 2$). This will let us observe the impact of hyperedge transmission as total infections dramatically increases.
The data are shown in Figure~\ref{fig:transmissions}.
For each epidemic trajectory, we compute the average transmissions by hyperedges of a given size over the last 1000 time steps (similar to Figure~\ref{fig:total-infections-explaination}). 
We then normalize this information to obtain a probability distribution for each value of $\alpha$ that represents the probability that a transmission was passed by a hyperedge of a given size. 
We find that initially, traditional pairwise edges are responsible for most transmissions but as the total number of infected nodes increases there is a transition to larger hyperedges transmitting infections.
For $0.3<\alpha<0.6$, transmissions primarily occur due to large hyperedges or small hyperedges. 
However, as we move to larger values of $1 \le \alpha \le 2$, medium-sized hyperedges become responsible for most transmissions.

In this case, we again see the utility of our model. Note the two regimes: for small $\alpha$, which cause a transition to an epidemic, we see that \emph{large hyperedges} matter more than medium size hyperedges. Whereas in a strong, steady state epidemic at large $\alpha$, then its the medium size hyperedges that matter. 

\begin{figure}[tp]
    \centering
    \begin{subfigure}[b]{0.49\textwidth}
        \centering
        \includegraphics[width=\textwidth]{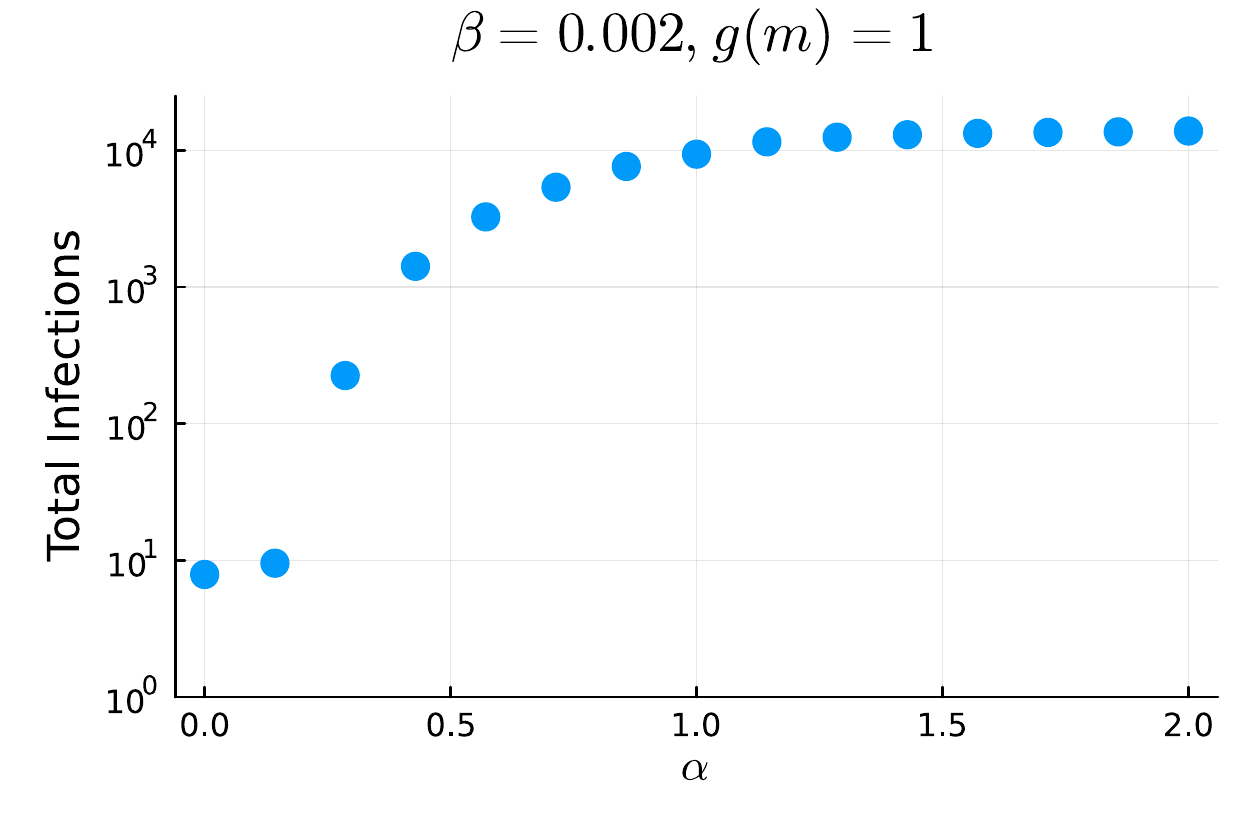}
        \label{fig:average-total-infections-2e-3}
    \end{subfigure}
    \begin{subfigure}[b]{0.49\textwidth}
        \centering
        \includegraphics[width=\textwidth]{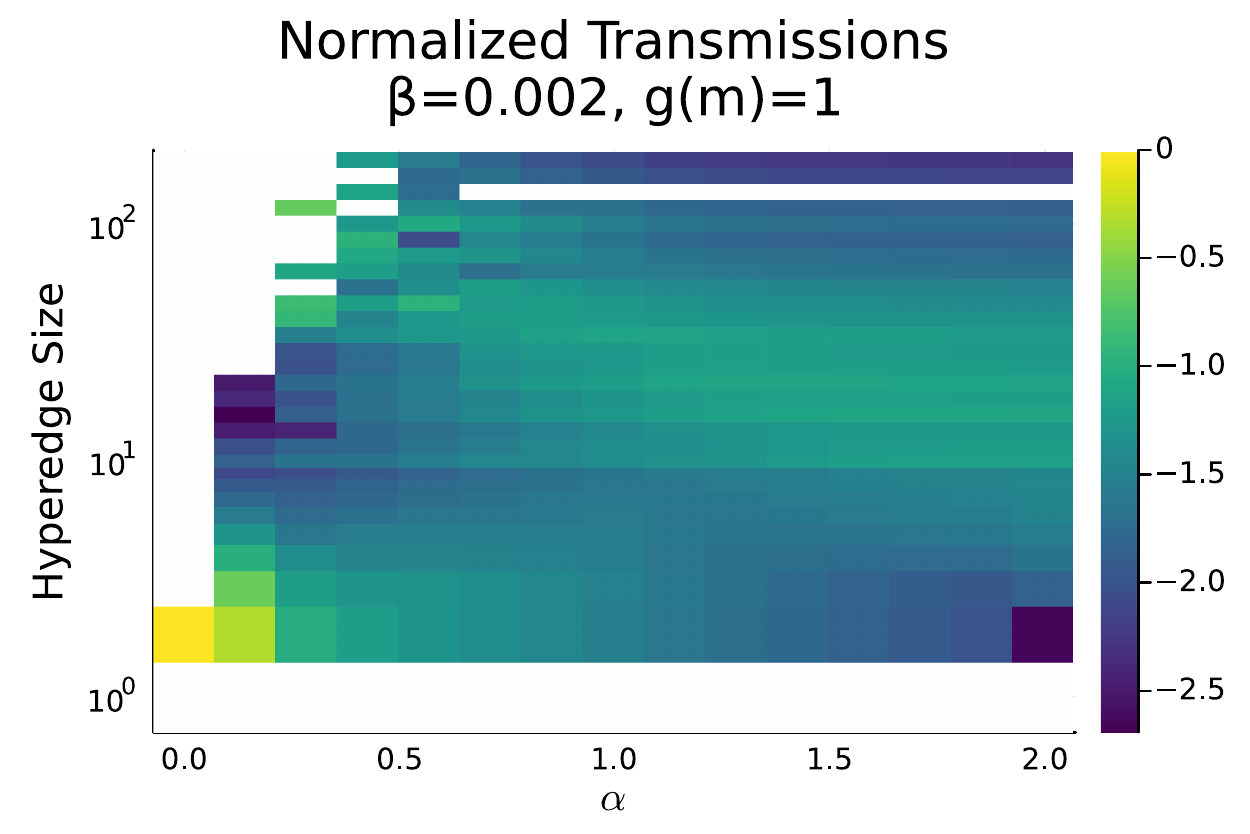}
        \label{fig:transmissions-2e-3}
    \end{subfigure}
    \caption{In a scenario without ventilation effects, higher order structure comes to dominant transmission. 
    (Left) As we increase $\alpha$, the underlying higher-order structure increases causing a dramatic increase in average trailing infected nodes in the $g(m) = 1$ case to model no ventilation. 
    (Right) We collect statistics on which hyperedge size was responsible for transmissions of an infection. We display the probability that a transmission occurred with a hyperedge of a given size in the steady state. At $\alpha = 0$, all infections are transmitted by pairwise edges because these are all that exist. For $\alpha$ a little bit larger (0.3 to 0.6), the larger hyperedges around $10^2$ dominate transmission right around the threshold. Then for $\alpha$ larger than around 1.0, we see a transition to medium size hyperedges dominating. This persists until $\alpha = 2$.}
    \label{fig:transmissions}
\end{figure}

\section{Discussion}
Pairwise interactions have enabled a large body of research with applications ranging from power grid robustness~\cite{amani2021power} to heterogeneity and disease spread~\cite{newman2002spread, pastor2015epidemic, moreno2002epidemic} to accelerating materials and molecule discovery~\cite{yang2024molecule}. 
The non-trivial effects of higher-order structure coupled with competing generalizations for intuitive pairwise concepts have posed a significant challenge for researchers with a grounded understanding in pairwise data~\cite{Benson-2016-motif-spectral,benson2021higherordernetworkanalysistakes}.
This manifests in a number of ways including inconclusive and contradictory findings depending on exactly how the higher-order problem is realized. We saw a key illustration of this in terms of the clustering coefficients (Section~\ref{sec:clustering}), where we investigated and saw a number of distinct properties of various generalizations of clustering coefficients to hypergraphs~\cite{robins2004small,yin2018higher,ha2024clustering,miyashita2024clustering}. As an example beyond what we have looked at, consider the impact of representation in the behavior of higher-order synchronization~\cite{zhang2023higher}. In this case, the choice of data representation (simplex versus hypergraph) can alter the behavior of synchronization.

The utility of our model is that we can isolate effects to the impact of the higher-order pieces from the rest of the graph model. We are optimistic that this will be broadly useful in the future to help understand the origin of ambiguities in the behavior of higher-order network models in the future.

There are many possible extensions of this model. For instance, we could consider extensions to \emph{gravity-like} model wherein the link or move around based on gravity-like terms with their neighbors~\cite{cabanastirapu2023humanmobilitydescribedclosedform,eldaghar2023multi}. In terms of epidemics, pairwise epidemics behave differently in the presence of interventions~\cite{eldaghar2023multi}, extending these types of interventions to hypergraphs would shed further insights into the role of higher-order spreading in more realistic scenarios. 
One avenue of future work we plan to investigate is \emph{fitting} our model to a dataset. In this case, we anticipate that a combination of modern vertex embedding techniques~\cite{Grover-2016-node2vec,Perozzi_2014,Tang_2015} will yield interesting models that are broadly similar to input networks.

\section*{Acknowledgments}
Eldaghar, Zhu, and Gleich all acknowledge funding and support from DOE award DE-SC0023162. Gleich is also partially supported by NSF IIS-2007481 and IARPA AGILE.

\newpage
\bibliographystyle{plain}
\bibliography{refs}

% \listoftodos

% \listoftodos
% \todo{replace "infections" with "infected nodes" in text}
% \todo{replace "infections" with "infected nodes" if possible in figure captions}
% \todo{include above listed citations}
% \todo{finish text for epidemic section}
% \todo{comment on duplication for figure 12 + 13 - one of the subplots is the same}

% \begin{itemize}
% \item epidemic figures
% \item see if we can find the three regimes with a fixed value of $\delta$. 
% \item 
% \item clustering: weighted from MartrixNetwork. 
% \item ppr: figures
% \item 
% \end{itemize}

% \paragraph{Notes}
% \begin{enumerate}
%     \item time history figure 
%     \begin{enumerate}
%         \item transmissions over time - set up for story 
%     \end{enumerate}
%     \item varying epidmeic params + ventilation - heatmaps (slide 23,24)  
%     \begin{enumerate}
%         \item hyperedges generally increase total infections w/o ventilation.
%         \item enough ventilation can quash this (depends on epidemic parameters)
%         \item and no bistability
%     \end{enumerate}
    
%     \item other hyperedge behavior 
%     \begin{enumerate}
%         \item when varying alpha, epidemic model matters. can get decrease over time, or increase or wave like thing depending on epidemic parameters
%         \item figure for spectral radii of pairwise projections
%         \item figure for weighted projected edge volumes
%     \end{enumerate}
%     \item hyperedge transmissions 
%     \begin{enumerate}
%         \item no ventilation 
%         \item sqrt ventilation 
%     \end{enumerate}
% \end{enumerate}

\end{document}